\documentclass[letter,11pt]{article}

\usepackage{setspace}
\usepackage{subcaption}
\usepackage{enumerate}
\usepackage{amsmath}
\usepackage{amsfonts}
\usepackage{graphicx}
\usepackage{tcolorbox}
\usepackage{amssymb}
\usepackage{multirow}
\usepackage{geometry}
\usepackage{soul}
\usepackage{colortbl}
\usepackage{wrapfig}
\usepackage{algorithm}
\usepackage{algorithmicx}
\usepackage{algpseudocode}
\usepackage{amsthm}
\usepackage{todonotes}
\usepackage{mathtools}
\usepackage{mathrsfs}

\usepackage[pdftex, plainpages = false, pdfpagelabels, 
                 bookmarks=false,
                 bookmarksopen = true,
                 bookmarksnumbered = true,
                 breaklinks = true,
                 linktocpage,
                 pagebackref,
                 colorlinks = true,  
                 linkcolor = blue,
                 urlcolor  = cyan,
                 citecolor = red,
                 anchorcolor = green,
                 hyperindex = true,
                 hyperfigures
                 ]{hyperref}

\usepackage{thmtools} 
\usepackage{thm-restate}
\usepackage{bm}

\algtext*{EndWhile}
\algtext*{EndIf}
\algtext*{EndFor}

\newtheorem{lemma}{Lemma}

\newtheorem{theorem}[lemma]{Theorem}
\newtheorem{theorem*}[lemma]{Theorem*}
\newtheorem{corollary}[lemma]{Corollary}

\newtheorem{observation}[lemma]{Observation}

\graphicspath{{./figures/}}

\title{\LARGE Algorithms for Halfplane Coverage and Related Problems\thanks{A preliminary version of this paper will appear in {\em Proceedings of the 40th International Symposium on Computational Geometry (SoCG 2024)}.}}
\author{Haitao Wang\footnote{University of Utah, Salt Lake City, UT 84112, USA. \texttt{haitao.wang@utah.edu}.} \and Jie Xue\footnote{New York University Shanghai, China. \texttt{jiexue@nyu.edu}.}}
\date{}

\def\bbR{\mathbb{R}}
\def\bbS{\mathbb{S}}

\def\calU{\mathcal{U}}
\def\calD{\mathcal{D}}
\def\calK{\mathcal{K}}
\def\calP{\mathcal{P}}
\def\hatS{\widehat{S}}

\begin{document}

\maketitle

\bibliographystyle{plainurl}

\begin{abstract}
Given in the plane a set of points and a set of halfplanes, we consider the problem of computing a smallest subset of halfplanes whose union covers all points. In this paper, we present an $O(n^{4/3}\log^{5/3}n\log^{O(1)}\log n)$-time algorithm for the problem, where $n$ is the total number of all points and halfplanes. 
This improves the previously best algorithm of $n^{10/3}2^{O(\log^*n)}$ time by roughly a quadratic factor. For the special case where all halfplanes are lower ones, our algorithm runs in $O(n\log n)$ time, which improves the previously best algorithm of $n^{4/3}2^{O(\log^*n)}$ time and matches an $\Omega(n\log n)$ lower bound. Further, our techniques can be extended to solve a star-shaped polygon coverage problem in $O(n\log n)$ time, which
in turn leads to an $O(n\log n)$-time algorithm for computing an instance-optimal $\epsilon$-kernel of a set of $n$ points in the plane. Agarwal and Har-Peled presented an $O(nk\log n)$-time algorithm for this problem in SoCG 2023, where $k$ is the size of the $\epsilon$-kernel; they also raised an open question whether the problem can be solved in $O(n\log n)$ time. Our result thus answers the open question affirmatively. 
\end{abstract}

\section{Introduction}
\label{sec:intro}

Let $P$ be a set of points and $H$ a set of halfplanes in the plane. 
In the {\em halfplane coverage} problem, the goal is to find a smallest subset of halfplanes in $H$ whose union covers all points in $P$. 
The problem was studied in the literature before.
Har-Peled and and Lee~\cite{ref:Har-PeledWe12} first proposed an $O(n^5)$-time algorithm for the problem, where $n=|P|+|H|$. 
Later, Pedersen and Wang~\cite{ref:PedersenAl22} provided an improved algorithm of $O(n^4\log n)$ time.
Very recently, Liu and Wang~\cite{ref:LiuOn23} solved the problem in $n^{10/3}2^{O(\log^*n)}$ time. 
In this paper, we present a new algorithm for the halfplane coverage problem of $O(n^{4/3} \cdot \log^{5/3}n\log^{O(1)}\log n)$ time, which improves the best known bound~\cite{ref:LiuOn23} by roughly a quadratic factor.   

In the {\em lower-only} version of the problem where all halfplanes in $H$ are lower-open (i.e., of the form $y \leq ax+b$), our new algorithm runs in $O(n\log n)$ time.
The lower-only halfplane coverage problem has also been studied before.
Chan and Grant's techniques~\cite{ref:ChanEx14} solved the problem in $O(n^4\log n)$ time.
Pedersen and Wang~\cite{ref:PedersenAl22} derived an $O(n^2\log n)$-time algorithm.
The recent work of Liu and Wang~\cite{ref:LiuOn23} proposed an improved solution of $n^{4/3}2^{O(\log^*n)}$ time, which was the best known result prior to this paper. We also show that $\Omega(n\log n)$ is a lower bound for the problem under the algebraic decision tree model; therefore, our algorithm is optimal. 

Interestingly, our new techniques for the halfplane coverage problem can also be applied to efficiently compute \textit{instance-optimal} kernels in the plane.
For a set $Q$ of $n$ points in the plane and a parameter $\epsilon\in (0,1)$, a subset $Q'\subseteq Q$ is an $\epsilon$-kernel of $Q$ if the projection of the convex hull of $Q'$ approximates that of $Q$ within $(1-\epsilon)$ factor in every direction.
Given $Q$ and $\epsilon$, the problem is to compute a smallest $\epsilon$-kernel for $Q$, called an {\em instance-optimal} $\epsilon$-kernel.
Very recently, Agarwal and Har-Peled~\cite{ref:AgarwalCo23} gave an $O(nk\log n)$-time algorithm for the problem in SoCG 2023, where $k$ is the size of the instance-optimal $\epsilon$-kernel.
They raised as an open question whether the problem can be solved in $O(n\log n)$ time.
We provide an $O(n\log n)$-time algorithm based on our techniques for halfplane coverage, and hence answer the open question affirmatively.

To see how the problem of computing instance-optimal $\epsilon$-kernels is related to halfplane coverage, we recall the approach of Agarwal and Har-Peled~\cite{ref:AgarwalCo23}.
In this approach, the task of finding instance-optimal $\epsilon$-kernel is first reduced in $O(n\log n)$ time to the following {\em star-shaped polygon coverage} problem.
Given a star-shaped polygon $R$ of $n$ vertices with respect to a center point $o$ (i.e., the segment $\overline{op}\subseteq R$ for every point $p\in R$) and a set of $n$ halfplanes that do not contain $o$, the problem is to compute a smallest subset of halfplanes whose union covers the boundary of $R$ (after the problem reduction, the size of a smallest subset is exactly equal to $k$, the size of an instance-optimal $\epsilon$-kernel in the original problem).
Agarwal and Har-Peled~\cite{ref:AgarwalCo23} then solved the star-shaped polygon coverage problem in $O(nk\log n)$ time.
By extending our techniques for the halfplane coverage problem, we present an improved algorithm of $O(n\log n)$ time for the star-shaped polygon coverage problem, which in turn solves the problem of computing instance-optimal $\epsilon$-kernels in $O(n\log n)$ time. 
In addition, we prove that $\Omega(n\log n)$ is a lower bound for the star-shaped polygon coverage problem under the algebraic decision tree model, and thus our algorithm is optimal.

Besides computing kernels, our techniques may find other applications as well.
One remarkable example is the computation of Hausdorff approximation~\cite{ref:AgarwalCo23,ref:Har-PeledOn23}.

\begin{itemize}
    \item Agarwal and Har-Peled~\cite{ref:AgarwalCo23} proposed an $O(nk\log n)$ time algorithm for computing an optimal Hausdorff approximation for a set of $n$ points in the plane, where $k$ is the size of the optimal solution.
    In this algorithm, they followed the similar techniques to the above instance-optimal $\epsilon$-kernel problem.
    Specifically, they reduced the Hausdorff approximation problem to a problem of covering the boundary of a star-shaped region $Q$ of curved boundary in the plane.
    Using our new techniques, the optimal Hausdorff approximation problem can now also be solved in $O(n\log n)$ time.

    \smallskip    
    \item Har-Peled and Raichel~\cite{ref:Har-PeledOn23} considered a ``dual'' Hausdorff approximation problem for a set of $n$ points in the plane and a parameter $k \in \{1,\dots,n\}$.
    They gave a randomized algorithm whose runtime is bounded by $O(\sqrt{k}(n\log n)^{3/2}+kn\log^2 n)$ with high probability. Their algorithm utilizes the above $O(nk\log n)$-time algorithm in~\cite{ref:AgarwalCo23} for the optimal Hausdorff approximation problem as a black box.
    Using our new $O(n\log n)$-time algorithm instead and follow the same analysis as in \cite{ref:Har-PeledOn23}, the dual Hausdorff approximation problem can now be solved in $O((n\log n)^{3/2})$ time with high probability.
\end{itemize}

\subsection{Other related work}
The halfplane coverage problem belongs to the domain of geometric set cover, which in turn is a special case of the (general) set cover problem.
For most types of geometric objects, the geometric set cover problem is NP-hard.
One example that has received much attention in the literature is the disk coverage problem, i.e., given a set of disks and a set of points in the plane, the problem is to find a smallest subset of disks that together cover all points.
The problem is NP-hard even if all disks have the same radius~\cite{ref:FederOp88}, while many approximation algorithms have been proposed~\cite{ref:LiA15,ref:MustafaSe14,ref:MustafaPt09}.
On the other hand, the problem becomes polynomial-time solvable if all disk centers lie on the same line or if the disk centers and the points are separated by a line~\cite{ref:AmbuhlCo06,ref:CarmiCo07,ref:ClaudeAn10,ref:LiuOn23,ref:PedersenAl22}.
Another well-studied case of geometric set cover is the rectangle coverage problem where the given geometric objects are (axis-parallel) rectangles.
This problem is also NP-hard, even when the rectangles are slabs, unit squares, or with certain constraints~\cite{ref:ChanEx14,ref:FowlerOp81}.
Agarwal and Pan~\cite{ref:AgarwalNe20} proposed an $O(\log\log k)$-approximation algorithm for rectangle coverage with near-linear running time, where $k$ is the optimum.
Finally, geometric set cover with other objective functions have also been studied, e.g., minimizing the membership or the ply of the solution~\cite{ref:BandyapadhyayMi23,ref:BiedlMi20,ref:DurocherMi23,ref:MitchellMi21}. 

For the $\epsilon$-kernel problem, Agarwal, Har-Peled, and Varadarajan proved in their seminal paper \cite{ref:AgarwalAp04} that an $\epsilon$-kernel of size $O(\epsilon^{-(d-1)/2})$ exists for a set of $n$ points in the $d$-dimensional space $\bbR^d$. However, an $\epsilon$-kernel could be much smaller in practice~\cite{ref:YuPr08}. 
Therefore, it is well-motivated to study algorithms for computing instance-optimal $\epsilon$-kernels.
The problem unfortunately becomes NP-hard in $\bbR^3$~\cite{ref:AgarwalCo23}.
Agarwal and Har-Peled~\cite{ref:AgarwalCo23} thus studied the exact algorithms for the 2D case. 

\subsection{Our approach}
For the lower-only halfplane coverage problem, we show that the problem can be reduced to an {\em interval coverage} problem: Given a set $P'$ of points and a set $S$ of intervals on the $x$-axis, compute a smallest subset of intervals whose union covers all points. The problem can be easily solved in $O(|P'|+|S|)$ time by a greedy algorithm after sorting. However, the issue with this approach is that while $|P'|=n$, we may have $|S|=\Omega(n^2)$ after the problem reduction. More specifically, points of $P'$ are defined by points of $P$ and intervals of $S$ are defined by halfplances of $H$.
While each point of $P$ defines a single point in $P'$ (thus $|P'|=n$), each halfplane could define $\Theta(n)$ intervals of $S$ (thus $|S|=\Omega(n^2)$). As such, the total time of the algorithm could be $\Omega(n^2)$. To improve the algorithm, our crucial observation is that it suffices to use one particular interval of $S$ for each halfplane. Consequently, we only need to use a subset $\hatS\subseteq S$ of size at most $n$ such that a smallest subset of $\hatS$ for covering $P'$ is also an optimal solution for covering $P'$ with $S$. In this way, the lower-only halfplane coverage problem is solved in $O(n\log n)$ time. 

For the star-shaped polygon coverage problem, we extend the above idea. We reduce the problem to a {\em circle coverage} problem: Given a set $S$ of arcs on a circle $C$, compute a smallest subset of arcs whose union covers the entire circle $C$. The problem can be solved in $O(|S|\log |S|)$ time~\cite{ref:AgarwalCo23,ref:LeeOn84}. As above, the issue is that $|S|$ could be $\Omega(n^2)$ (more specifically, arcs of $S$ are defined by halfplanes of $H$ and each halfplane may define $\Theta(n)$ arcs). Indeed, this is the main obstacle that prevents Agarwal and Har-Peled~\cite{ref:AgarwalCo23} from obtaining an $O(n\log n)$ time algorithm. Our new and crucial observation is that, as above, it is sufficient to use only one particular arc for each halfplane. Consequently, we only need to use a subset $\hatS\subseteq S$ of size at most $n$ such that a smallest subset of arcs of $\hatS$ covering $C$ is also an optimal solution for covering $C$ with $S$. In this way, the star-shaped polygon coverage problem is solved in $O(n\log n)$ time. 

To solve the general halfplane coverage problem, we consider two cases depending on whether $H$ has three halfplanes whose union is $\bbR^2$. In the case where $H$ does not have such three halfplanes, by Helly's theorem, the common intersection of the complements of all halfplanes of $H$ is not empty; let $o$ be a point in the common intersection. Then, no halfplane of $H$ contains $o$. With the help of $o$, we can compute in $O(n\log n)$ time a smallest subset of halfplanes to cover all points, by extending the idea for solving the lower-only halfplane coverage problem. 
In the case where $H$ has three halfplanes whose union is $\bbR^2$, the optimal solution size $\tau^*$ is at most three. If $\tau^*=3$, then it suffices to find three halfplanes from $H$ whose union is $\bbR^2$, which can be done in linear time~\cite{ref:Har-PeledWe12}. If $\tau^*=1$, then this case can be easily solved in $O(n\log n)$ time using halfplane range emptiness queries. For the remaining case $\tau^*=2$, we wish to find two halfplanes from $H$ whose union covers all points of $P$. Although this looks like a special case, it turns out this is the bottleneck case of the entire problem, which is surprising (and perhaps also interesting). Our algorithm for this case runs in $O(n^{4/3}\log^{5/3}n\log^{O(1)}\log n)$ time, which takes significantly more time than all other parts of the overall algorithm (all other parts together take $O(n\log n)$ time). Although we do not have a proof, we feel that $\Omega(n^{4/3})$ might be a lower bound for this subproblem (and thus the entire problem), at least under a somewhat restricted computational model~\cite{ref:EricksonNe96}.

\paragraph{Remark.}
After the submission of the conference version of this paper, Liu and Wang~\cite{ref:LiuOn24} gave a new implementation of their original algorithm in \cite{ref:LiuOn23} for the lower-only halfplane coverage problem; the new implementation runs in $O(n\log n)$ time. Although their algorithm also reduces the problem to an interval coverage problem, it is fundamentally different from ours. For instance, their algorithm first identifies some ``prunable'' halfplanes that are useless to the optimal solution and then proceed to reduce to the interval coverage problem using the remaining halfplanes. Our algorithm, in contrast, does not need such a pruning step. In addition, it seems difficult to extend their algorithm to the circle coverage or polygon coverage problem because it relies on the left-to-right order of all points. As such, our techniques appear to be more powerful. On the other hand, their approach can be used to solve a more general line-separable unit-disk coverage problem (the lower-only halfplane coverage is just a special case of the problem). Also, using their $O(n\log n)$ time new implementation for the lower-only halfplane coverage, their algorithm for the general halfplane coverage now runs in $O(n^3\log n)$ time. 

\paragraph{Outline.} The rest of the paper is organized as follows.
In Section~\ref{sec:lower}, we present our algorithm for the lower-only halfplane coverage problem. By extending the techniques, we solve the star-shaped polygon coverage problem in Section~\ref{sec:star}. As mentioned above, with the $O(n\log n)$-time problem reduction in~\cite{ref:AgarwalCo23}, this also solves the 2D instance-optimal $\epsilon$-kernel problem in $O(n\log n)$ time. The general halfplane coverage problem is discussed in Section~\ref{sec:general} and our approach uses an algorithm similar to the lower-only case algorithm as a subroutine. 

\section{Lower-only halfplane coverage}
\label{sec:lower}

In this section, we present an $O(n\log n)$-time algorithm for the lower-only halfplane coverage problem. Let $P$ be a set of points and $H$ a set of lower halfplanes in $\bbR^2$. We wish to compute a smallest subset of halfplanes whose union covers $P$ (i.e., covers all points of $P$). To simplify the notation, let $n=|P|=|H|$. 

\subsection{Preliminaries}

We call a subset of $H$ a {\em feasible solution} if the halfplanes of the subset together cover $P$; an {\em optimal solution} refers to a smallest feasible solution. 
We assume that the union of halfplanes of $H$ covers $P$ since otherwise there would be no feasible solution. Indeed, this can be easily determined $O(n\log n)$ time, e.g., by first computing the upper envelope of the bounding lines of all halfplanes of $H$ and then check whether every point of $P$ is below the upper envelope.

We make a general position assumption that no two points of $P$ have the same $x$-coordinate. We also assume that no two halfplanes of $H$ have their bounding lines parallel (since otherwise the one with lower bounding line is redundant and can be removed from $H$ because it is completely contained in the other halfplane).

We sort all points of $P$ from left to right and let $p_1,p_2,\ldots,p_n$ be the sorted list. For any two indices $i,j$ with $1\leq i\leq j\leq n$, let $P[i,j]$ denote the subsequence of points $\{p_i,p_{i+1},\ldots,p_j\}$.

For any halfplane $h\in H$, let $\ell_h$ denote the bounding line of $h$. 

Consider a halfplane $h\in H$. 
A subsequence $P[i,j]$ of $P$, with $1\leq i\leq j\leq n$, is called a {\em maximal subsequence covered by $h$} if all points of $P[i,j]$ are in $h$ but neither $p_{i-1}$ nor $p_{j+1}$ is in $h$ (to make the definition rigorous, we could add two dummy points $p_0$ and $p_{n+1}$ that are not covered by any halfplane of $H$). Let $\Gamma_h$ denote the set of all maximal subsequences of $P$ covered by $h$. It is not difficult to see that subsquences of $\Gamma_h$ are pairwise disjoint. 

\subsection{Reducing to interval coverage}
\label{sub:reduction}

We reduce the problem to an instance of the interval coverage problem on a set $P'$ of points and a set $S$ of line segments (or intervals) on the $x$-axis. For each point $p_i\in P$, let $p_i'$ be the (vertical) projection of $p_i$ onto the $x$-axis (i.e., $p_i'$ has the same $x$-coordinate as $p_i$). Define $P'=\{p_i'\ |\ 1\leq i\leq n\}$. As such, $P'$ has exactly $n$ points. For any $1\leq i\leq j\leq n$, let $P'[i,j]=\{p_i',p_{i+1}',\ldots,p_j'\}$. 

We now define the set $S$. For each halfplane $h\in H$, we define a set $S_h$ of segments as follows. 
For each subsequence $P[i,j]\in \Gamma_h$, we create a segment on the $x$-axis, denoted by $s[i,j]$, whose left and right endpoints are $p_i'$ and $p_j'$, respectively; we add $s[i,j]$ to $S_h$. As such, $P'\cap s[i,j]=P'[i,j]$. We say that $s[i,j]$ is defined by $h$. Define $S=\bigcup_{h\in H}S_h$.

Consider the following interval coverage problem on $P'$ and $S$: Find a smallest subset of segments of $S$ whose union covers $P'$. This problem can be solved by a simple greedy algorithm in $O(|P'|+|S|)$ time after sorting the left endpoints of all segments of $S$ along with all points of $P'$. 
Let $S^*\subseteq S$ be an optimal solution to the above interval coverage problem. Based on $S^*$, we create a subset $H^*\subseteq H$ as follows. For each segment $s\in S^*$, if $s$ is defined by a halfplane $h\in H$, then we add $h$ to $H^*$ (it is possible that the same segment $s$ is defined by multiple halfplanes, in which case we add an arbitrary such halfplane to $H^*$). In what follows, we show that $H^*$ is an optimal solution to our original halfplane coverage problem for $P$ and $H$. To this end, we first prove the following lemma. 

\begin{lemma}\label{lem:10}
\begin{enumerate}
    \item The union of all halfplanes of $H^*$ covers $P$.   
    \item $P'$ can be covered by $k$ segments of $S$ if and only if $P$ can be covered by $k$ halfplanes of $H$. 
\end{enumerate}
\end{lemma}
\begin{proof}
For any point $p_i\in P$, since $p_i'$ must be covered by a segment $s[i,j]$ of $S^*$, by definition, $p_i$ must be covered by a halfplane $h$ that defines $s[i,j]$ and $h$ is in $H^*$. Therefore, the first lemma statement holds.  

We now prove the second lemma statement. 
Suppose $S$ has a subset $S'$ of $k$ segments whose union covers $P'$. Let $H'$ be the subset of defining halfplanes of all segments of $S'$. By an argument similar to the first lemma statement, the halfplanes of $H'$ together cover $P$. Since $|H'|\leq |S'|=k$, we obtain that $P$ can be covered by $k$ halfplanes of $H$. 

On the other hand, suppose $H$ has a subset $H'$ of $k$ segments whose union covers $P$. We create a subset $S'\subseteq S$ of at most $k$ halfplanes that together cover $P'$ as follows. Let $\calU$ be the upper envelope of the bounding lines of all halfplanes of $S'$. 
Since the union of the halfplanes of $H'$ covers $P$, every point of $P$ is below $\calU$. Note that the bounding line of each halfplane of $H'$ contains at most one edge of $\calU$. For each edge $e$ of $\calU$, let $h$ be the halfplane whose bounding line contains $e$. As $e$ is an edge of $\calU$, no point of $P$ is vertically above $e$. Therefore, the vertical projection $e'$ of $e$ onto the $x$-axis, which is a segment, must be contained in at most one segment $s[i,j]$ of $S_h$ and we add $s[i,j]$ to $S'$ (note that such a segment $s[i,j]\in S_h$ must exist if $P$ has a point vertically below $e$). Since the bounding line of each halfplane contains at most one edge of $\calU$, the size of $S'$ is at most $k$. We claim that the union of all segments of $S'$ must cover $P'$. Indeed, consider any point $p_i'\in P'$. $\calU$ must has an edge $e$ that is vertically above $p_i$. Therefore, the halfplane $h$ must define a segment $s[i,j]$ in $S_h$ that contains $e'$, where $h$ is the halfplane whose bounding line contains $e$. By definition, $S'$ must contain $s[i,j]$. This proves that the union of all segments of $S'$ must cover $P'$.

The second lemma statement thus follows. 
\end{proof}

The above lemma immediately leads to the following corollary. 

\begin{corollary}\label{coro:lower}
\begin{enumerate}
    \item The size of a smallest subset of $H$ for covering $P$ is equal to the size of a smallest subset of $S$ for covering $P'$.     
    \item No two segments of $S^*$ are defined by the same halfplane.
    \item $H^*$ is a smallest subset of $H$ for covering $P$. 
\end{enumerate}
\end{corollary}
\begin{proof}
The first statement immediately follows from Lemma~\ref{lem:10}(2). 

For the second statement, assume to the contrary that $S^*$ has two segments $s$ and $s'$ that are defined by the same halfplace $h\in H$. Then, by the definition of $H^*$, $|H^*|<|S^*|$ holds. As $S^*$ is a smallest subset of $S$ for covering $P'$, and $H^*$ is a feasible solution for covering $P$ by Lemma~\ref{lem:10}(1), we obtain that the size of a smallest subset of $H$ for covering $P$ is smaller than the size of a smallest subset of $S$ for covering $P'$. But this contradicts the first corollary statement. 

For the third statement, since $|H^*|\leq |S^*|$, $S^*$ is a smallest subset of $S$ for covering $P'$, and $H^*$ is a feasible solution for covering $P$, by the first corollary statement, $|H^*|=|S^*|$ and $H^*$ must be a smallest subset of $H$ for covering $P$. 
\end{proof}

In light of Corollary~\ref{coro:lower}, the above gives an algorithm that computes an optimal solution to the halfplane coverage problem for $P$ and $H$. However, the algorithm is not efficient because the size of $S$ could be $\Omega(n^2)$. Indeed, it is not difficult to see that $|\Gamma_h|$ (and thus $|S_h|$) could be $\Theta(n)$ for each halfplane $h\in H$. Hence, $|S|$ could be $\Omega(n^2)$ in the worst case. In the following, we reduce the time to $O(n\log n)$ by showing that a smallest subset of $S$ for covering $P'$ can be computed in $O(n\log n)$ time by using only a small subset of $S$. 

\subsection{Improvement}

The main idea is to show that it suffices to use at most one segment from $S_h$ for each halfplane $h\in H$. More specifically, we show that for each halfplane $h\in H$, it is sufficient to define a segment, denoted by $s(h)$, for at most one subsequence of $\Gamma_h$, such that $\widehat{S}=\{s(h)\ |\ h\in H\}$, which is of size at most $n$ and is a subset of $S$, must contain a smallest subset of $S$ for covering $P'$. This implies that to compute a smallest subset of $S$ for covering $P'$, it suffices to compute a smallest subset of $\hatS$ to cover $P'$. The latter problem can be solved faster as $|\hatS|\leq n$. 

In the following, we first define the segment $s(h)$ for each halfplane $h\in H$. Then, we prove that $\hatS$ contains a smallest subset of $S$ for covering $P'$. Finally, we discuss how to compute $\hatS$ efficiently. 

\paragraph{Defining $\boldsymbol{s(h)}$ and $\boldsymbol{\hatS}$.}
For each halfplane $h\in H$, we define a segment $s(h)$ as follows. Let $\calU$ denote the upper envelope of the bounding lines of all halfplanes of $H$. Depending on whether the bounding line $\ell_h$ of $h$ contains an edge on $\calU$, there are two cases. 
\begin{enumerate}
    \item If $\ell_h$ contains an edge $e$ on $\calU$, then let $e'$ be the vertical projection of $e$ on the $x$-axis (see Fig.~\ref{fig:lower10}). Since $e$ is on $\calU$, no point of $P$ is vertically above $e$. Hence, $e'$ is contained in at most one segment $s[i,j]$ of $S_h$ (note that such a segment may not exist, e.g., if $e'$ is between $p_i'$ and $p_{i+1}'$ and neither point is in $h$). If such a segment $s[i,j]$ exists, then we define $s(h)=s[i,j]$; otherwise, $s(h)$ is not defined. 

    \item If $\ell_h$ does not contain any edge on $\calU$, then let $v_h$ be the unique vertex of $\calU$ that has a tangent line parallel to $\ell_h$ (see Fig.~\ref{fig:lower20}). Let $v_h'$ be the vertical projection of $v_h$ onto the $x$-axis. Clearly, $S_h$ has at most one segment $s[i,j]$ containing $v_h'$. If such a segment $s[i,j]$ exists, then we define $s(h)=s[i,j]$; otherwise, $s(h)$ is not defined. 
\end{enumerate}

\begin{figure}[t]
\begin{minipage}[t]{\textwidth}
\begin{center}
\includegraphics[height=1.4in]{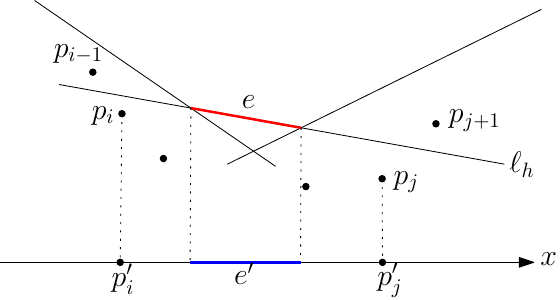}
\caption{\footnotesize Illustrating the definition of $s(h)=s[i,j]$ for the case where $\ell_h$ contains an edge $e$ of $\calU$.}
\label{fig:lower10}
\end{center}
\end{minipage}
\vspace{-0.1in}
\end{figure}

\begin{figure}[t]
\begin{minipage}[t]{\textwidth}
\begin{center}
\includegraphics[height=1.4in]{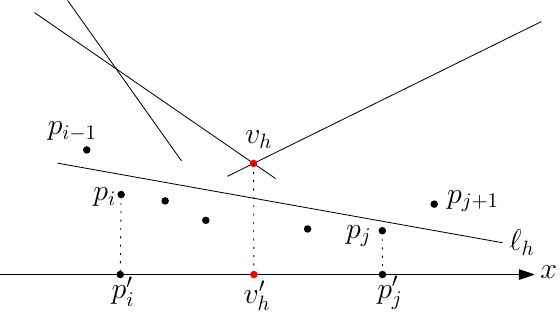}
\caption{\footnotesize Illustrating the definition of $s(h)=s[i,j]$ for the case where $\ell_h$ does not contain any edge of $\calU$.}
\label{fig:lower20}
\end{center}
\end{minipage}
\vspace{-0.1in}
\end{figure}

Define $\hatS=\{s(h)\ |\ h\in H\}$. 

The following implies that a smallest subset of segments of $\hatS$ whose union covers $P'$ is an optimal solution to the interval coverage problem for $P'$ and $S$. 

\begin{lemma}\label{lem:30}
For any segment $s\in S\setminus \hatS$, $\hatS$ has a segment $s'$ containing $s$.
\end{lemma}
\begin{proof}
Consider a segment $s[i,j]\in S\setminus \hatS$. Let $h$ be the halfplane that defines $s[i,j]$. Depending on whether the bounding line $\ell_h$ contains an edge of $\calU$, there are two cases. 

\paragraph{$\boldsymbol{\ell_h}$ contains an edge of $\boldsymbol{\calU}$.}
Suppose $\ell_h$ contains an edge $e$ of $\calU$. Then, let $e'$ be the vertical projection of $e$ onto the $x$-axis. Since $e$ is an edge of $\calU$, $e'$ must be contained in at most one segment of $S_h$. Hence, either $e'$ is disjoint from $s[i,j]$ or $e'\subseteq s[i,j]$. In the latter case, by definition, we have $s(h)=s[i,j]$ and thus $s[i,j]\in \hatS$, but this is not possible since $s[i,j]\not\in \hatS$. As such, $s[i,j]$ must be disjoint from $e'$. Without loss of generality, we assume that $s[i,j]$ is left of $e'$ (see Fig.~\ref{fig:lowerproof10}). 

By definition, all points of $P[i,j]$ are in $h$ and thus below $\ell_h$. Consider the point $p_j$, which is the rightmost point of $P[i,j]$. Let $q$ be the intersection of $\calU$ with the vertical line through $p_j$ and let $h'$ be the halfplane whose bounding line contains $q$ (see Fig.~\ref{fig:lowerproof10}). Clearly, $q$ is to the left of $e$. Therefore, the slope of $\ell_{h'}$ must be smaller than that of $\ell_{h}$. We claim that $s[i,j]\subseteq s(h')$ must hold, which will prove the lemma since $s(h')$ is in $\hatS$.
To this end, suppose $s(h')$ is defined by a subsequence $P[i',j']$ of $P$. Then, it suffices to show that $P[i,j]\subseteq P[i',j']$. Indeed, by the definition of $q$, since both $\ell_{h'}$ and $\ell_{h}$ contain edges of $\calU$ and the slope of $\ell_{h'}$ is smaller than that of $\ell_{h}$, the intersection $\ell_{h'}\cap \ell_{h}$ must be to the right of $q$. Therefore, for any point $p$ to the left of $q$ and below $\ell_{h}$, $p$ must be below $\ell_{h'}$ as well. Since all points of $P[i,j]$, which are to the left of $q$, are in $h$ and thus below $\ell_h$, they must be below $\ell_{h'}$ and thus are in $h'$. In addition, since $q\in \ell_{h'}$ is on $\calU$ and $q$ is the intersection of $\calU$ with the vertical line through $p_j$, by definition, we have $p_j\in P[i',j']$. This implies that $P[i',j']$ is the subsequence of $\Gamma_{h'}$ that contains $p_j$. Since all points of $P[i,j]$ are in $h'$, it must hold that $P[i,j]\subseteq P[i',j']$. This proves that $s[i,j]\subseteq s(h')$, which also proves the lemma.   

\begin{figure}[t]
\begin{minipage}[t]{\textwidth}
\begin{center}
\includegraphics[height=1.6in]{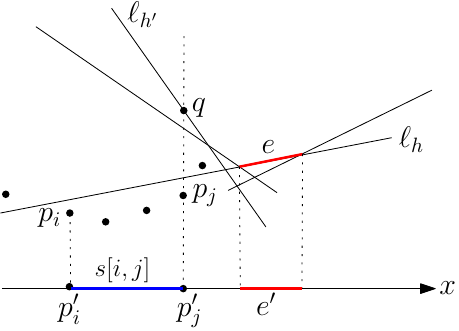}
\caption{\footnotesize Illustrating the proof of Lemma~\ref{lem:30} when $\ell_h$ contains an edge of $\calU$.}
\label{fig:lowerproof10}
\end{center}
\end{minipage}
\vspace{-0.1in}
\end{figure}


\paragraph{$\boldsymbol{\ell_h}$ does not contain any edge of $\boldsymbol{\calU}$.}
Suppose $\ell_h$ does not contain any edge of $\calU$. Let $v_h$ be the unique vertex of $\calU$ that has a tangent line parallel to $\ell_h$ (see Fig.~\ref{fig:lowerproof20}). Let $v_h'$ be the vertical projection of $v_h$ on the $x$-axis. Note that $v_h'$ cannot be in $s[i,j]$ since otherwise $s[i,j]$ would be $s(h)$ and thus in $\hatS$, incurring contradiction. Without loss of generality, we assume that $s[i,j]$ is left of $p'$. 

\begin{figure}[t]
\begin{minipage}[t]{\textwidth}
\begin{center}
\includegraphics[height=1.6in]{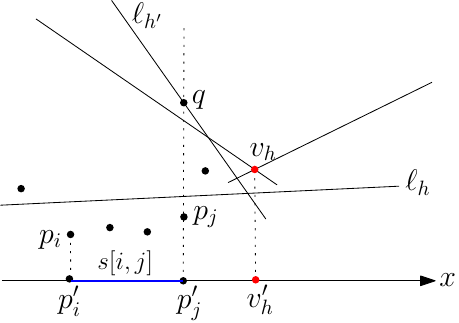}
\caption{\footnotesize Illustrating the proof of Lemma~\ref{lem:30} when $\ell_h$ does not contain any edge of $\calU$.}
\label{fig:lowerproof20}
\end{center}
\end{minipage}
\end{figure}

Define $q$ and $h'$ in the same way as above (see Fig.~\ref{fig:lowerproof20}). We claim that $s[i,j]\subseteq s(h')$. The argument follows exactly the same way as above by using the property that the slope of $\ell_h'$ is smaller than that of $\ell_h$.
This proves the lemma as $s(h')$ is in $\hatS$.   
\end{proof}

With Lemma~\ref{lem:30}, the lower-only halfplane coverage problem on $P$ and $H$ can be solved as follows. (1) Compute $\hatS$. (2) Compute a smallest subset $\hatS^*$ of segments of $\hatS$ whose union covers $P'$. (3) Using $\hatS^*$, obtain an optimal solution for $P$ and $H$ (in the same way as described in Section~\ref{sub:reduction}). Since $|\hatS|\leq n$ and $|P'|=n$, the second and third steps can be implemented in $O(n\log n)$ time. The following lemma shows that the first step can be done in $O(n\log n)$ time too, by making use of ray-shooting queries in simple polygons~\cite{ref:ChazelleRa94,ref:ChazelleVi89,ref:HershbergerA95}. 

\begin{lemma}
Computing all segments of $\hatS$ can be done in $O(n\log n)$ time. 
\end{lemma}
\begin{proof}
To compute $\hatS$, we manage to reduce the problem to ray-shooting queries in simple polygons~\cite{ref:ChazelleRa94,ref:ChazelleVi89,ref:HershbergerA95}.  To this end, we first construct a simple polygon $R$ of $O(n)$ vertices as follows. 

Let $p_0'$ be a point to the left of $p_1'$ and $p_{n+1}'$ a point to the right of $p_n'$ on the $x$-axis. Let $R'$ be an axis-parallel rectangle with $\overline{p_0'p_{n+1}'}$ as a bottom edge and containing all points of $P$. We add the segments $\overline{p_ip_i'}$, for all $1\leq i\leq n$, to $R'$ to obtain a (weakly) simple polygon $R$ (see Fig.~\ref{fig:polygonlower}). Clearly, $R$ has $O(n)$ vertices. Note that each edge $e$ of $R$ belongs to one of five cases: (1) $e$ is the top edge of $R'$; (2) $e$ is the left edge of $R'$; (3) $e$ is the right edge of $R'$; (4) $e$ is a segment $\overline{p_ip_i'}$ for some $1\leq i\leq n$; (5) $e$ is $\overline{p_{i}'p_{i+1}'}$ for some $0\leq i\leq n$. 

\begin{figure}[t]
\begin{minipage}[t]{\textwidth}
\begin{center}
\includegraphics[height=0.8in]{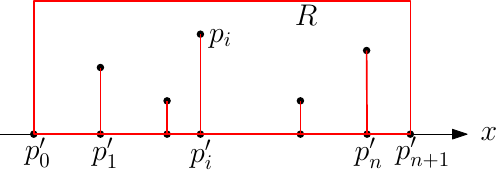}
\caption{\footnotesize Illustrating the (weakly) simple polygon $R$, which is formed by all red segments.}
\label{fig:polygonlower}
\end{center}
\end{minipage}
\vspace{-0.1in}
\end{figure}

We build a ray-shooting data structure for $R$ in $O(n)$ time so that given a ray originating from a point inside $R$, the first edge of $R$ hit by the ray can be computed in $O(\log n)$ time~\cite{ref:ChazelleRa94,ref:ChazelleVi89,ref:HershbergerA95}. With the help of the ray-shooting data structure, we next compute $\hatS$ as follows.

We start with computing the upper envelope $\calU$ of the bounding lines of all halfplanes of $H$, which takes $O(n\log n)$ time. 
For each edge $e$ of $\calU$, suppose $e$ is on the bounding line of a halfplane $h$. To compute $s(h)$, suppose $s(h)$ is defined by a subsequence $P[i',j']$ of $\Gamma_h$. We can compute the index $i'$ as follows. We shoot a ray originating from any interior point of $e$ toward the left endpoint of $e$. By a ray-shooting query using the ray-shooting data structure, we can find the first edge $e'$ of $R$ hit by the ray in $O(\log n)$ time. As discussed above, $R$ has five types of edges. If $e'$ is the top or left edge of $R'$, then $i'=1$. Notice that $e'$ cannot be the right edge of $R'$. If $e'$ is a segment $\overline{p_ip_i'}$, for some $1\leq i\leq n$, then $i'=i+1$. If $e'$ is $\overline{p_i'p_{i+1}'}$, for some $1\leq i\leq n$, then $i'=i+1$. In this way, the index $i'$ can be determined in $O(\log n)$ time. The index $j'$ can be computed analogously (by shooting a rightward ray from an interior point of $e$). Note that if it turns out that $i'>j'$, then $s(h)$ is not defined. As such, the segment $s(h)$ can be computed in $O(\log n)$ time. 

Now consider a halfplane $h$ whose bounding line $\ell_h$ does not contain any edge of $\calU$. We find the vertex $v_h$ of $\calU$ as defined before, which can be done in $O(\log n)$ time by binary search on $\calU$. Let $v_h''$ be the intersection of $\ell_h$ and the vertical line through $v_h$.
To compute $s(h)$, suppose $s(h)$ is defined by a subsequence $P[i',j']$ of $\Gamma_h$. In the same way as above, we can compute the two indices $i'$ and $j'$ in $O(\log n)$ time by shooting two rays from $v_h''$ along $\ell_h$ with opposite directions. 

In summary, the segment $s(h)$ for each halfplane $h\in H$ can be computed in $O(\log n)$ time. Therefore, computing the set $\hatS$ can be done in $O(n\log n)$ time in total. 
\end{proof}

We thus obtain the following result. 

\begin{theorem}
Given in the plane a set of points and a set of lower halfplanes, one can compute a smallest subset of halfplanes whose union covers all points in $O(n\log n)$ time, where $n$ is the total number of all points and halfplanes. 
\end{theorem}


\subsection{Lower bound}
The following proves an $\Omega(n\log n)$ time lower bound even for some special cases. 
\begin{theorem}\label{theo:lb-lowerhalfplane}
It requires $\Omega(n\log n)$ time to solve the lower halfplane coverage problem under the algebraic decision tree model, even if all points are given sorted by $x$-coordinates or if the bounding lines of all halfplanes are given sorted by their slopes. 
\end{theorem}
\begin{proof}
Let $P$ be a set of $n$ points and $H$ a set of $n$ lower halfplanes. 
If all points of $P$ are given sorted by their $x$-coordinates, we refer to it as the {\em point-sorted case}. If halfplanes of $H$ are given sorted by the slopes of their bounding lines, we refer to it as the {\em slope-sorted case}. 
We show below that both cases have an $\Omega(n\log n)$ lower bound. 

\paragraph{The point-sorted case.}
We begin with the point-sorted case. We use a reduction from the set equality problem. 

Let $A=\{a_1,a_2,\ldots,a_n\}$ and $B=\{b_1,b_2,\ldots,b_n\}$ be two sets of $n$ numbers each, such that the numbers of $A$ are distinct and sorted. The {\em set equality} problem is to determine whether $A=B$. The problem requires $\Omega(n\log n)$ time to solve under the algebraic decision tree model~\cite{ref:Ben-OrLo83}. Next, we reduce the problem to the point-sorted case. 
We create an instance of the point-sorted case in $O(n)$ time as follows. 

For each number $a_i\in A$, we create a point $p_i$ with coordinate $(a_i,a_i^2)$, which is the point of the graph $y=x^2$ whose $x$-coordinate is equal to $a_i$. Let $P$ be the set of all points $p_i$ thus created. Since numbers of $A$ are already sorted, points of $P$ are also sorted by their $x$-coordinates. 
For each number $b_j\in B$, we create a lower halfplane $h_j$ whose bounding line is tangent to the graph $y=x^2$ at $(b_j,b_j^2)$. Let $H$ be the set of all lower halfplanes $h_j$ thus created. 
Notice that $h_j$ covers a point $p_i\in P$ if and only if $a_i=b_j$. Since numbers of $A$ are distinct, the size of a smallest subset of $H$ covering $P$ is equal to $n$ if and only if $A=B$. This completes the reduction (which takes $O(n)$ time).

\paragraph{The slope-sorted case.}
For the slope-sorted case, we use a reduction from the set inclusion problem. 

Given $A$ and $B$ in the same way as above, the {\em set inclusion} problem is to determine whether $B\subseteq A$. The problem requires $\Omega(n\log n)$ time to solve under the algebraic decision tree model~\cite{ref:Ben-OrLo83}. To reduce the problem to the slope-sorted case, we can simply follow the above approach but switch the roles of $A$ and $B$, i.e., use $A$ to create $H$ and use $B$ to create $P$. Since $A$ is sorted, $H$ is also sorted by the slopes of the bounding lines of the halfplanes. 
Observe that the lower-halfplane coverage problem for $P$ and $H$ has a solution if and only if $B\subseteq A$. This completes the reduction (which takes $O(n)$ time).
\end{proof}


\section{Star-shaped polygon coverage and 2D instance-optimal $\epsilon$-kernels}
\label{sec:star}
In this section, we solve the star-shaped polygon coverage problem. Let $\calP$ be a star-shaped polygon with respect to a center point $o$, i.e., the line segment $\overline{op}\subseteq \calP$ for any point $p\in \calP$. Let $n$ be the number of vertices of $\calP$. Let $H$ be a set of $n$ halfplanes that do not contain $o$. 
The goal is to compute a smallest subset of halfplanes whose union covers $\partial \calP$, the boundary of $\calP$. We present an $O(n\log n)$ time algorithm for the problem. 

For any halfplane $h$, denote by $\overline{h}$ the complement halfplane of $h$. Define $\overline{H}=\{\overline{h}\ |\ h\in H\}$. 

As in Section~\ref{sec:lower}, we assume that no two halfplanes of $H$ have their bounding lines parallel to each other. 
We also assume that $H$ has a feasible solution (i.e., $H$ has a subset of halfplanes whose union covers $\partial P$). Indeed, this can be determined in $O(n\log n)$ time as follows. We first compute the common intersection $\calU$ of the halfplanes of $\overline{H}$, which must contain $o$. $\calU$ can be computed in $O(n\log n)$ time. Notice that $H$ has a feasible solution if and only if the interior of $\calU$ does not intersect $\partial \calP$~\cite{ref:AgarwalCo23}. Since $\calP$ is star-shaped with respect to $o$, whether the interior of $\calU$ intersects $\partial \calP$ can be determined in $O(n)$ by rotating a sweeping ray round $o$. As such, in total $O(n\log n)$ time, one can determine whether $H$ has a feasible solution. 
In the following, we focus on finding an optimal solution. 

Let $C$ be a circle centered at $o$ and containing $\calP$. For any point $p$ inside $C$ with $p\neq o$, define its {\em projection point on $C$} as the intersection between $C$ and the ray from $o$ to $p$. 

\subsection{Reducing to circle coverage}

Following a similar idea to that in the previous section, we reduce the problem to a {\em circle coverage} problem on $C$: Given a set of circular arcs on $C$, find a smallest subset of arcs that together cover $C$.

Consider a halfplane $h$. Recall that $\ell_h$ denote its bounding line. Let $\Gamma_h$ denote the set of the maximal segments of $\calP\cap \ell_h$ (intuitively $\Gamma_h$ plays a similar role to the same notation in Section~\ref{sec:lower}). For each segment $s\in \Gamma_h$, it ``blocks'' the visibility of a portion of $\partial \calP$ from $o$, and it also blocks the visibility of an arc of $C$ from $o$, denoted by $\alpha_s$ (see Fig.~\ref{fig:star00}). The two endpoints of $\alpha_s$ are exactly the projections of the two endpoints of $s$ on $C$, respectively. Let $S_h$ denote the set of arcs defined by the segments of $\Gamma_h$. Define $S=\bigcup_{h\in H}S_h$. 

\begin{figure}[t]
\begin{minipage}[t]{\textwidth}
\begin{center}
\includegraphics[height=1.8in]{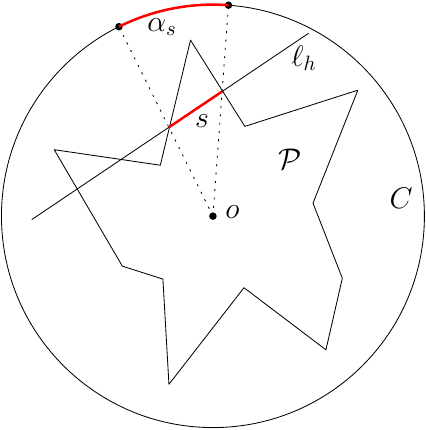}
\caption{\footnotesize Illustrating the definition of an arc $\alpha_s$.}
\label{fig:star00}
\end{center}
\end{minipage}
\vspace{-0.1in}
\end{figure}

Now consider the circle coverage problem for $S$: Compute a smallest subset of arcs of $S$ whose union covers $C$. To solve the problem, Lee and Lee~\cite{ref:LeeOn84} first gave an $O(|S|\log |S|)$ time algorithm, and Agarwal and Har-Peled~\cite{ref:AgarwalCo23} presented a much simpler solution with the same runtime. 
Suppose we have an optimal solution $S^*$. We create a subset $H^*$ of $H$ as follows. For each arc $\alpha$ of $S^*$, if $h$ is the halfplance that defines $\alpha$, then we include $h$ in $H^*$. 
The following lemma, which is analogous to Lemma~\ref{lem:10} in Section~\ref{sec:lower}, has been proved by Agarwal and Har-Peled~\cite{ref:AgarwalCo23}.

\begin{lemma}\label{lem:60}{\em (Agarwal and Har-Peled~\cite{ref:AgarwalCo23})}
\begin{enumerate}
    \item The union of all halfplanes $H^*$ covers $\partial \calP$.   
    \item $C$ can be covered by $k$ arcs of $S$ if and only if $\partial \calP$ can be covered by $k$ halfplanes of $H$. 
\end{enumerate}
\end{lemma}

With the above lemma, we can obtain results analogous to Corollary~\ref{coro:lower}. 
In particular, $H^*$ is a smallest subset of $H$ for covering $\partial \calP$.
As such, 
the above gives an algorithm for computing a smallest subset of $H$ to cover $\partial \calP$. However, the algorithm is not efficient because the size of $S$ could be $\Omega(n^2)$ (since $|\Gamma_h|$ and thus $|S_h|$ could be $\Theta(n)$ for each halfplane $h\in H$). In the following, we reduce the time to $O(n\log n)$ by showing that a smallest subset of $S$ for covering $\partial \calP$ can be computed in $O(n\log n)$ time by using only a small subset of $S$.  

\subsection{Improvement}

As for the lower-only halfplane coverage problem, we will define a subset $\hatS\subseteq S$ such that $\hatS$ contains at most one arc $\alpha(h)$ defined by each halfplane $h\in H$ (and thus $|\hatS|\leq n$) and $\hatS$ contains a smallest subset of $S$ for covering $C$. Further, we will show that $\hatS$ can be computed in $O(n\log n)$ time. Consequently, applying the circle coverage algorithm~\cite{ref:AgarwalCo23,ref:LeeOn84} can solve the problem in $O(n\log n)$ time. 

\paragraph{Defining $\boldsymbol{\alpha(h)}$ and $\boldsymbol{\hatS}$.}
For each halfplane $h\in H$, we define an arc $\alpha(h)$ on $C$ as follows. As will be seen, $\alpha(h)$ is in $S_h$ and thus is in $S$. Let $\calU$ be the common intersection of all halfplanes of $\overline{H}$. As discussed before, $\calU$ must be inside $\calP$ since $H$ has a feasible solution. 

We represent a direction in $\bbR^2$ by a unit vector. Let $\bbS$ denote the set of all unit vectors (directions) in $\bbR^2$, i.e., $\bbS=\{v: \lVert v\rVert=1\}$.

For each halfplane $h$, we define the {\em norm} of its bounding line $\ell_h$ as the direction perpendicular to $\ell_h$ and towards the interior of $h$. For each edge $e$ of $\calU$, it is contained in $\ell_h$ for some halfplane $h\in H$; the norm of $e$ refers to the norm of $\ell_h$. 

The norms of all edges of $\calU$ partition $\bbS$ into $|\calU|$ {\em basic intervals} such that the interior of each interval does not contain the norm of any edge of $\calU$, where $|\calU|$ refers to the number of edges of $\calU$. 
Note that the endpoints of every basic interval are norms of two adjacent edges of $\calU$.

To define $\alpha(h)$, depending on whether the bounding line $\ell_h$ of $h$ contains an edge of $\calU$, there are two cases. 

\begin{enumerate}
\item If $\ell_h$ contains an edge $e$ of $\calU$, then since $\calU$ is inside $\calP$, $e$ must be contained in a segment $s$ of $\Gamma_h$. Let $\alpha(h)$ be the arc of $C$ defined by $s$ (see Fig.~\ref{fig:starcase10}).
\item If $\ell_h$ does not contain any edge of $\calU$, then let $I$ be the basic interval of $\bbS$ that contains the norm of $h$. Let $e_1$ and $e_2$ be the two adjacent edges of $\calU$ whose norms are endpoints of $I$. Define $v_h$ to be the vertex of $\calU$ incident to both $e_1$ and $e_2$ (see Fig.~\ref{fig:starcase20}). Let $v_h'$ be the projection of $v_h$ on $C$. Define $\alpha(h)$ to be the unique arc (if exists) of $S_h$ containing $v'_h$. 
\end{enumerate}

\begin{figure}[t]
\begin{minipage}[t]{\textwidth}
\begin{center}
\includegraphics[height=1.8in]{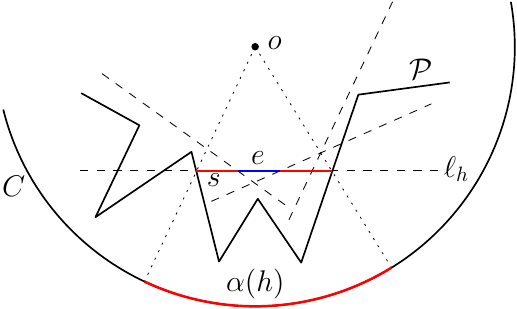}
\caption{\footnotesize Illustrating the definition of $\alpha(h)$ when $\ell_h$ contains an edge $e$ of $\calU$. The dashed lines are halfplane bounding lines. The solid polyline is part of $\partial \calP$. The red arc on the circle $C$ is $\alpha(h)$.}
\label{fig:starcase10}
\end{center}
\end{minipage}
\vspace{-0.1in}
\end{figure}

\begin{figure}[t]
\begin{minipage}[t]{\textwidth}
\begin{center}
\includegraphics[height=1.8in]{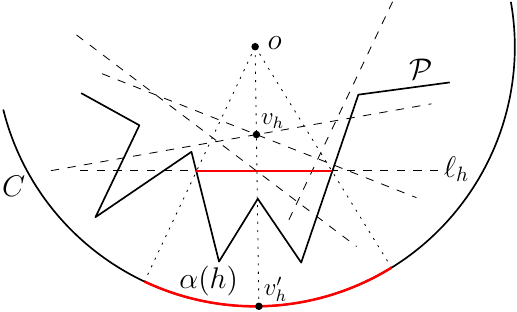}
\caption{\footnotesize Illustrating the definition of $\alpha(h)$ when $\ell_h$ does not contain any edge of $\calU$. The red arc on the circle $C$ is $\alpha(h)$.}
\label{fig:starcase20}
\end{center}
\end{minipage}
\vspace{-0.1in}
\end{figure}

Define $\hatS=\{s(h)\ |\ h\in H\}$. 

The following lemma implies that a smallest subset of arcs of $\hatS$ whose union covers $C$ is also a smallest subset of arcs of $S$ covering $C$. 

\begin{lemma}\label{lem:70}
For any arc $\alpha\in S\setminus \hatS$, $\hatS$ must have an arc $\alpha'$ such that $\alpha\subseteq \alpha'$.
\end{lemma}
\begin{proof}
Consider an arc $\alpha\in S\setminus \hatS$. Let $h$ be the halfplane that defines $\alpha$. By definition, $\alpha$ is the projection of a segment of $\Gamma_h$ on $C$; let $s$ denote the segment. 

Without loss of generality, we assume that $\ell_h$ is horizontal and $o$ is above $\ell_h$. Depending on whether $\ell_h$ contains an edge of $\calU$, there are two cases. 

\paragraph{$\boldsymbol{\ell_h}$ contains an edge of $\boldsymbol{\calU}$.}
Suppose $\ell_h$ contains an edge $e$ of $\calU$.
Then $\Gamma_h$ has a segment $s_e$ that contains $e$. Let $\alpha_e$ be the arc of $S_h$ defined by $s_e$. By definition, $\alpha(h)$ is $\alpha_e$ and thus $\alpha_e$ is in $\hatS$. Since $\alpha\not\in \hatS$, we have $\alpha\neq \alpha_e$. This implies that $s\neq s_e$. Without loss of generality, we assume that $s$ is to the left of $s_e$ and thus $s$ is to the left of $e$ (see Fig.~\ref{fig:starproof10}). 

\begin{figure}[t]
\begin{minipage}[t]{\textwidth}
\begin{center}
\includegraphics[height=1.8in]{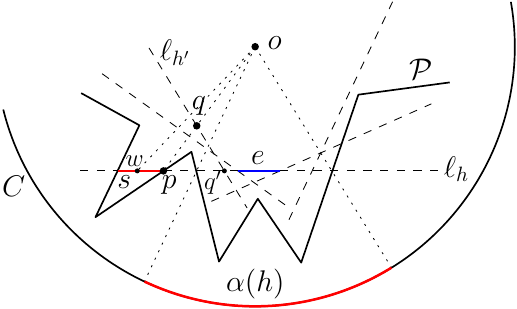}
\caption{\footnotesize Illustrating the proof of Lemma~\ref{lem:70} for the case where $\ell_h$ contains an edge $e$ of $\calU$.}
\label{fig:starproof10}
\end{center}
\end{minipage}
\vspace{-0.1in}
\end{figure}

Let $p$ be the right endpoint of $s$ (see Fig.~\ref{fig:starproof10}). 
Since $s$ is left of $e$, which is an edge of $\calU$, $s$ is completely outside $\calU$. 
Hence, $\overline{op}$ intersects $\partial \calU$ at a point, denoted by $q$. Let $e'$ be the edge of $\calU$ containing $q$. Let $h'$ be the halfplane whose bounding line $\ell_{h'}$ containing $e'$. In the following, we argue that $\alpha\subseteq\alpha(h')$ must hold, which will prove the lemma as $\alpha(h')\in \hatS$. To this end, it suffices to show that for any point $w\in s$, $\overline{ow}$ intersects $s(e')$, where $s(e')$ is the segment of $\Gamma_{h'}$ containing $e'$. Indeed, since $e$ is an edge of $\calU$, $s$ is left of $e$, and $\overline{op}$ intersects $e'$ at $q$, we obtain that $\ell_{h'}$ must intersect $\ell_h$ at a point $q'$ between $s$ and $e$ (see Fig.~\ref{fig:starproof10}). Since $p$ is the right endpoint of $s$ and $w\in s$, it follows that $\overline{ow}$ must intersect $\ell_{h'}$. Further, since $w\in \calP$, $\overline{ow}\cap \ell_{h'}$ must be in $\calP$. As $w$ is an arbitrary point of $s$, the above argument essentially shows that the portion $s'$ of $\ell_{h'}$ that consists of the intersection points $\overline{ow}\cap \ell_{h'}$ for all points $w\in s$ must be inside $\calP$. Note that $s'$ contains $q$ since when $w=p$, we have $q=\overline{ow}\cap \ell_{h'}$. As $q\in e'$, by definition, we have $s'\subseteq s(e')$. As such, we obtain that $\overline{ow}$ must intersect $s(e')$ for any point $w\in s$. This proves that $\alpha\subseteq \alpha(h')$ and thus proves the lemma. 

\paragraph{$\boldsymbol{\ell_h}$ does not contain any edge of $\boldsymbol{\calU}$.}
If $\ell_h$ does not contain any edge of $\calU$, then let $I$ be the basic interval of $\bbS$ that contains the norm of $\ell_h$. 
Recall the definitions of $v_h$ and $v'_h$. 
By definition, $\alpha(h)$ is the unique arc (if exists) of $S_h$ containing $v_h'$. Since $\alpha\not\in \hatS$ and $\alpha(h)\in \hatS$, we obtain that $\alpha\neq \alpha(h)$ and thus $\alpha$ does not contain $v_h'$. This implies that $s$ does not contain $v_h''$, where $v_h''=\ell_h\cap \overline{ov_h'}$. Without loss of generality, we assume that $s$ lies to the left of $v_h''$ (see Fig.~\ref{fig:starproof20}). 

\begin{figure}[t]
\begin{minipage}[t]{\textwidth}
\begin{center}
\includegraphics[height=1.8in]{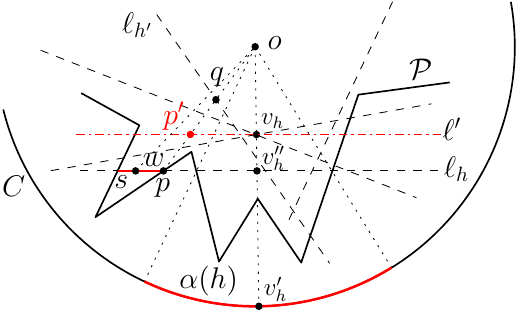}
\caption{\footnotesize Illustrating the proof of Lemma~\ref{lem:70} for the case where $\ell_h$ does not contain any edge of $\calU$.}
\label{fig:starproof20}
\end{center}
\end{minipage}
\vspace{-0.1in}
\end{figure}

Let $p$ be the right endpoint of $s$ (see Fig.~\ref{fig:starproof20}). As $\ell_h$ does not contain an edge of $\calU$, $\ell_h$ is completely outside $\calU$. Hence, $\overline{op}$ must intersect $\partial \calU$ at a point $q$. Let $e'$ be the edge of $\calU$ containing $q$ and let $h'$ be the halfplane whose bounding line $\ell_{h'}$ contains $e'$. In the following, we argue that $\alpha\subseteq\alpha(h')$ must hold, which will prove the lemma as $\alpha(h')\in \hatS$. To this end, as in the above case, it suffices to show that for any point $w\in s$, $\overline{ow}$ must intersect $s(e')$, where $s(e')$ is the segment of $\Gamma_{h'}$ containing $e'$. Indeed, by definition, the line through $v_h$ parallel to $\ell_h$ must be tangent to $\calU$ at $v_h$; let $\ell'$ denote the line. Let $p'$ be the intersection of $\ell'$ and $\overline{op}$. Hence, $q=\overline{op'}\cap e'$ and $p'$ is left of $v_h$. Since $o$ is inside $\calU$ and $\ell'$ is tangent to $\calU$ at $v_h$, $\ell_{h'}$ must intersect $\ell'$ at a point right of $p'$. This means that $\ell_{h'}$ must intersect $\ell_h$ at a point right of $p$. Consequently, for any point $w\in s$, which is left of $p$, $\overline{ow}$ must intersect $\ell_{h'}$. Further, since $w\in s\subseteq \calP$, the intersection $\overline{ow}\cap \ell_{h'}$ must be inside $\calP$.
The rest of the argument is the same as the above case. 
Specifically, as $w$ is an arbitrary point of $s$, the above argument essentially shows that the portion $s'$ of $\ell_{h'}$ that consists of the intersection points $\overline{ow}\cap \ell_{h'}$ for all points $w\in s$ must be inside $\calP$. Note that $s'$ contains $q$ since when $w=p$, we have $q=\overline{ow}\cap \ell_{h'}$. As $q\in e'$, by definition, we have $s'\subseteq s(e')$. As such, we obtain that $\overline{ow}$ must intersect $s(e')$ for any point $w\in s$. 
This proves that $\alpha\subseteq \alpha(h')$ and thus proves the lemma. 
\end{proof}


With Lemma~\ref{lem:70}, a smallest subset of $H$ for covering $\partial \calP$ can be computed as follows. (1) Compute $\hatS$. (2) Compute a smallest subset $\hatS^*$ of $\hatS$ for covering $C$. (3) Using $\hatS^*$, obtain a smallest subset of $H$ to cover $\partial \calP$. Since $|\hatS|\leq n$, the second and third steps can be done in $O(n\log n)$ time using the circle coverage algorithm in \cite{ref:AgarwalCo23,ref:LeeOn84}. The following lemma shows that the first step can be done in $O(n\log n)$ time too, using the ray-shooting queries in simple polygons~\cite{ref:ChazelleRa94,ref:ChazelleVi89,ref:HershbergerA95}. 

\begin{lemma}
Computing all arcs of $\hatS$ can be done in $O(n\log n)$ time. 
\end{lemma}
\begin{proof}
First of all, we can find a circle $C$ in $O(n)$ time. 
We then compute $\calU$ in $O(n\log n)$ time. Next, we build a ray-shooting data structure for $\calP$ in $O(n)$ time so that given any ray with origin inside $\calP$, the first edge of $\calP$ hit by the ray can be computed in $O(\log n)$ time~\cite{ref:ChazelleRa94,ref:ChazelleVi89,ref:HershbergerA95}. For each halfplane $h\in H$, we compute the arc $\alpha(h)$ as follows. 

If $\ell_h$ contains an edge $e$ of $\calU$, then the maximal segment $s$ of $\ell_h\cap \calP$ containing $e$ can be computed in $O(\log n)$ time by two ray-shooting queries from any point of $e$ (following the two directions parallel to $\ell_h$). With the segment $s$, $\alpha(h)$ can be determined in additional $O(1)$ time as $\alpha(h)$ is the projection of $s$ on $C$.

If $\ell_h$ does not contain any edge of $\calU$, we first find the vertex $v_h$ by binary search on $\calU$. Let $v_h''$ be the intersection of $\ell_h$ and the ray from $o$ to $v_h$. By definition, the arc $\alpha(h)$ is the projection of $s$ on $C$, where $s$ is the maximal segment of $\ell_h\cap \calP$ containing $v_h''$. Hence, it suffices to compute $s$, which can be done in $O(\log n)$ time by two ray-shooting queries from $v_h''$ (following the two directions parallel to $\ell_h$).

As such, $\alpha(h)$ for each halfplane $h\in H$ can be computed in $O(\log n)$ time and therefore the total time for computing $\hatS$ is $O(n\log n)$.
\end{proof}

The following theorem summarizes our result. 

\begin{theorem}\label{theo:star}
The star-shaped polygon coverage problem is solvable in $O(n\log n)$ time, where $n$ is the sum of the number of vertices of the polygon and the number of halfplanes.     
\end{theorem}


\paragraph{Computing instance-optimal $\epsilon$-kernels in $\bbR^2$.}
As discussed in Section~\ref{sec:intro}, computing an instance-optimal $\epsilon$-kernel for a set of $n$ points in the plane can be reduced in $O(n\log n)$ time to an instance of the star-shaped polygon coverage problem with $n$ halfplanes and a star-shaped polygon of $n$ vertices. Consequently, with our algorithm for the star-shaped polygon coverage problem, an instance-optimal $\epsilon$-kernel can be computed in $O(n\log n)$ time. 

\paragraph{Covering an $x$-monotone polyline.} A special case of the star-shaped polygon coverage problem is as follows. Given in the plane an $x$-monotone polyline $\calP$ of $n$ vertices and a set $H$ of $n$ lower halfplances, we aim to compute a smallest subset of halfplanes so that their union covers $\calP$. If we consider a point $o$ with $y$-coordinate at $+\infty$, then the problem becomes a special case of the star-shaped polygon coverage problem and thus can be solved in $O(n\log n)$ time. 

\subsection{Lower bound}
We finally prove the lower bound below, which justifies the optimality of Theorem~\ref{theo:star}. 
\begin{theorem}
Solving the star-shaped polygon coverage problem requires $\Omega(n\log n)$ time under the algebraic decision tree model. 
\end{theorem}
\begin{proof}
We prove an $\Omega(n\log n)$ lower bound for the $x$-monotone polyline coverage problem, which imlies the same lower bound for the more general star-shaped polygon coverage problem.
We use a reduction from the point-sorted case of the lower halfplane coverage problem, which has an $\Omega(n\log n)$ time lower bound under the algebraic decision tree model as proved in Theorem~\ref{theo:lb-lowerhalfplane}. 

Let $P$ be a set of $n$ points and $H$ a set of $n$ lower halfplanes, such that points of $P=\{p_1,p_2,\ldots,p_n\}$ are  given sorted by their $x$-coordinates. We create an instance of the $x$-monotone polyline coverage problem in $O(n)$ time as follows. 

\begin{figure}[t]
\begin{minipage}[t]{\textwidth}
\begin{center}
\includegraphics[height=0.8in]{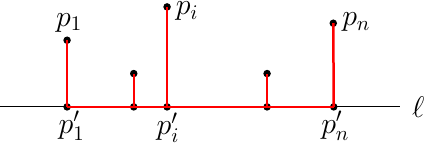}
\caption{\footnotesize Illustrating the (weakly) polyline $\calP$, which is formed by all red segments.}
\label{fig:polyline}
\end{center}
\end{minipage}
\end{figure}

First of all, we use the same set $H$ as the set of lower halfplanes in the polyline coverage problem. Next we create an $x$-monotone polyline $\calP$ based on the points of $P$. Let $\ell$ be a horizontal line below all points of $P$ and with $y$-coordinate equal to $y_{\ell}$ to be determined later. For each point $p_i\in P$, let $p_i'$ be the vertical projection of $p_i$ on $\ell$. Define $\calP$ as the (weakly) $x$-monotone polyline formed by the following segments (see Fig.~\ref{fig:polyline}): $\overline{p_ip_i'}$ for all $1\leq i\leq n$, and $\overline{p'_{i}p'_{i+1}}$ for all $1\leq i\leq n-1$. 

It remains to determine the $y$-coordinate $y_{\ell}$. The main idea is to make the line $\ell$ low enough (i.e., make $y_{\ell}$ small enough) so that whenever a halfplane $h$ of $H$ covers a point $p_i$, $h$ must cover both $\overline{p_ip_{i+1}}$ and $\overline{p_{i-1}p_{i}}$ (if $i=1$, only covering $\overline{p_{i}p_{i+1}}$ is necessary; if $i=n$, only covering $\overline{p_{i-1}p_{i}}$ is necessary). To achieve the effect, we can determine $y_{\ell}$ as follows. Let $\kappa$ be the largest absolute slope value of all halfplane bounding lines of $H$ (note that $\kappa$ is a finite value since no halfplane bounding line is vertical). 
Consider a segment $\overline{p_ip_{i+1}}$ for some $1\leq i\leq n-1$. Based on the $y$-coordinate of $p_i$ and the value $\kappa$, we can easily calculate a largest value $y_1$ (resp., $y_2$) such that for any lower halfplane $h$ (not necessarily in $H$) with absolute slope value equal to $\kappa$ and covering $p_i$ (resp., $p_{i+1}$), $h$ must cover $\overline{p_ip_{i+1}}$ if $y_{\ell}$ were equal to $y_1$ (resp., $y_2$). Let $y(i,i+1)=\min\{y_1,y_2\}$. As such, for any lower halfplane $h$ (not necessarily in $H$) whose  absolute slope value is at most $\kappa$ such that $h$ covers either $p_i$ or $p_{i+1}$, $h$ must cover $\overline{p_ip_{i+1}}$ if $y_{\ell}$ were equal to $y(i,i+1)$. Now set $y_{\ell}=\min_{1\leq i\leq n-1}y(i,i+1)$. According to the above discussion, our sought effect can be achieved with this value $y_{\ell}$. With the effect, it is not difficule to see that a subset of $H$ covering $P$ if and only if the same subset covers the polyline $\calP$. This completes the reduction (which takes $O(n)$ time). 
\end{proof}

\section{The general halfplane coverage}
\label{sec:general}

In this section, we consider the general halfplane coverage problem. Given in the plane a set $P$ of $n$ points and a set $H$ of $n$ halfplanes, the goal is to compute a smallest subset of halfplanes so that their union covers all points of $P$. Note that $H$ may have both upper and lower halfplanes. We present an $O(n^{4/3}\log^{5/3}n\log\log^{O(1)})$ time algorithm for the problem. 

For a halfplane $h$, denote by $\overline{h}$ the complement halfplane of $h$. Let $\overline{H}=\{\overline{h}\ |\ h\in H\}$. 

We first determine whether the union of all halfplanes of $H$ is the entire plane. This can be done by computing the common intersection $\calU$ of all halfplanes of $\overline{H}$, which can be done in $O(n\log n)$ time. Notice that $\calU=\emptyset$ if and only if the union of all halfplanes of $H$ is the entire plane. Depending on whether $\calU=\emptyset$, our algorithm will proceed in different ways. In what follows, we first discuss the case $\calU\neq \emptyset$ in Section~\ref{sec:nonempty} and the other case is solved in Section~\ref{sec:empty}.

\subsection{The case $\boldsymbol{\calU\neq \emptyset}$}
\label{sec:nonempty}
Assuming that $\calU\neq \emptyset$, we solve this case in $O(n\log n)$ by an algorithm similar to those in the previous two sections. One may consider it a {\em cyclic version} of the lower-only halfplane coverage algorithm in Section~\ref{sec:lower} or a {\em point version} of the start-shaped polygon coverage algorithm in Section~\ref{sec:star}.  
Since $\calU\neq\emptyset$, let $o$ be any point inside $\calU$. Then, no halfplanes of $H$ cover $o$. Let $C$ be a circle containing $o$ and all points of $P$. 

We reduce the problem to a {\em circular-point coverage problem}: Given on $C$ a set $P'$ of points and a set $S$ of arcs, the goal is to compute a subset of arcs whose union covers all points. This problem is different from the circle coverage problem in Section~\ref{sec:star} in that here we only need to cover points of $P'$ instead of the entire circle $C$. To solve this new problem, we will show that the problem can be reduced to the circle coverage problem and consequently applying the circle coverage algorithm~\cite{ref:AgarwalCo23,ref:LeeOn84} can solve the problem.

\subsubsection{Reducing to the circular-point coverage problem}

We order the points of $P$ counterclockwise around $o$ and let $p_1,p_2,\ldots,p_n$ be the ordered list. For any point $p$ in $C$ with $p\neq o$, we define its projection point on $C$ as the intersection between $C$ and the ray from $o$ to $p$. For each point $p_i\in P$, let $p_i'$ be its projection on $C$. Hence, the points $p'_1,p'_2,\ldots,p'_n$ are ordered on $C$ counterclockwise. Let $P'$ be the set of all these projection points. 

We consider $P'$ as a cyclic sequence of points. Each point $p'_i\in P'$ refers to $p_{j}'$ with $j=i \mod n$ if $i>n$. For two indices $i,j$, we use $P'[i,j]$ to denote the subsequence of points $p_i',p_{i+1}',\ldots,p'_j$. We define the same notation for the points of $P$. 

\begin{figure}[t]
\begin{minipage}[t]{\textwidth}
\begin{center}
\includegraphics[height=2.0in]{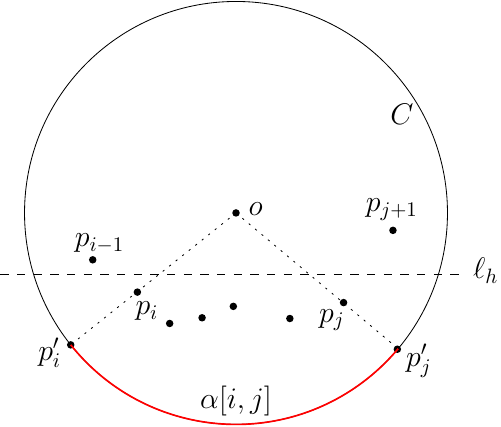}
\caption{\footnotesize Illustrating the definition of arcs $\alpha[i,j]$.}
\label{fig:general00}
\end{center}
\end{minipage}
\vspace{-0.1in}
\end{figure}

For each halfplane $h\in H$, we define a set $S_h$ of arcs on $C$ as follows. 
A subsequence $P[i,j]$ of $P$ is called a {\em maximal subsequence covered by $h$} if all points of $P[i,j]$ are in $h$ but neither $p_{i-1}$ or $p_{j+1}$ is in $h$. Let $\Gamma_h$ denote the set of all maximal subsequences of $P$ covered by $h$. For each subsequence $P[i,j]$ of $\Gamma_h$, we define an arc $\alpha[i,j]$ on $C$ starting from $p_i'$ counterclockwise until $p_j'$ (see Fig.~\ref{fig:general00}). Let  $S_h$ be the set of all arcs thus defined by the subsequences of $\Gamma_h$. Define $S=\bigcup_{h\in H}S_h$.

Now consider the circular-point coverage problem for $S$ and $P'$ on $C$: Compute a smallest subset of arcs of $S$ whose union covers $P'$. We will discuss later that the problem can be solved in $O((|S|+|P'|)\log(|S|+|P'|))$ time. Suppose we have an optimal solution $S^*$. We create a subset $H^*$ of $H$ as follows. For each arc $\alpha$ of $S^*$, if $\alpha$ is defined by a halfplane $h\in H$, then we add $h$ to $H^*$. We have the following lemma, analogous to Lemmas \ref{lem:10} and \ref{lem:60}. Its proofs are also similar and we include it here for completeness. 

\begin{lemma}\label{lem:100}
\begin{enumerate}
    \item The union of all halfplanes of $H^*$ covers $P$.   
    \item $P'$ can be covered by $k$ arcs of $S$ if and only if $P$ can be covered by $k$ halfplanes of $H$. 
\end{enumerate}
\end{lemma}
\begin{proof}
For any point $p_i\in P$, since $p_i'$ must be covered by an arc $\alpha$ of $S^*$, by definition, $p_i$ must be covered by the halfplane $h$ that defines $\alpha$ and $h$ is in $H^*$. As such, the first lemma statement follows. 

We next prove the second lemma statement. 

Suppose $S$ has a subset $S'$ of $k$ arcs that together cover $P'$. Let $H'$ be the subset of halfplanes that define the arcs of $S'$. Hence, $|H'|\leq k$. Using the same argument as that for the first 
lemma statement, the union of all halfplanes of $H'$ must cover $P$. As such, $P$ can be covered by $k$ halfplanes of $H$. 

On the other hand, suppose $H$ has a subset $H'$ of $k$ halfplanes that together cover $P$. We create a subset $S'\subseteq S$ of at most $k$ arcs that form a coverage of $P'$, as follows. 

Let $\calU'$ be the common intersection of the complements of the halfplanes of $H'$. Note that $\calU'\neq \emptyset$ (e.g., $o\in \calU'$). Since the union of halfplanes of $H'$ covers $P$, we have $\calU'\cap P=\emptyset$. For each edge $e$ of $\calU'$, let $h$ be the halfplane whose bounding line contains $e$. Let $e'$ be the arc of $C$ that is the projection of $e$ on $C$ (i.e., $e'$ consists of the projections of all points of $e$ on $C$). As $e$ is an edge of $\calU'$, no point of $P$ is in the triangle formed by $o$ and $e$, implying that points of $P'$ covered by $e'$ must form a subsequence of $P'[i,j]$ of $P'$. This means that $e'$ must be contained in at most one arc $\alpha$ of $S_h$ and such an arc must exist if $e'$ covers a point of $P'$. We add $\alpha$ to $S'$ if $\alpha$ exists. As such, since the bounding line of each halfplane contains at most one edge of $\calU'$, the size of $S'$ is at most $k$. We next show that the arcs of $S'$ together form a coverage of $P'$. 

Consider a point $p_i'\in P'$. By the definition of $\calU'$, $\calU'$ must have an edge $e$ such that its projection $e'$ on $C$ covers $p_i'$. According to the above discussion, $h$ must defines an arc $\alpha$ in $S_h$ that contains $e'$, where $h$ is the halfplane whose bounding line contains $e$. By the definition of $S'$, $\alpha$ must be in $S'$. This proves that the union of all arcs of $S'$ covers $P'$. 

The second lemma statement thus follows. 
\end{proof}

With Lemma~\ref{lem:100}, we can obtain results analogous to Corollary~\ref{coro:lower}. 
In particular, $H^*$ is an optimal solution of $H$ for covering $P$.

The above gives an algorithm for computing a smallest subset of $H$ for covering $P$. However, the algorithm is not efficient because $|S|$ could be $\Omega(n^2)$ (since $|\Gamma_h|$ and thus $|S_h|$ could be $\Theta(n)$ for each halfplane $h\in H$). In the following, we reduce the time to $O(n\log n)$ by showing that a smallest subset of $S$ for covering $P'$ can be computed in $O(n\log n)$ time by using only a small subset of $S$. 

\subsubsection{Improvement}

As for the lower-only halfplane coverage problem, we will define a subset $\hatS\subseteq S$ such that $\hatS$ contains at most one arc $\alpha(h)$ defined by each halfplane $h\in H$ (and thus $|\hatS|\leq n$) and $\hatS$ contains a smallest subset of $S$ for covering $P'$. Further, we will show that $\hatS$ can be computed in $O(n\log n)$ time. 

\paragraph{Defining $\boldsymbol{\alpha(h)}$ and $\boldsymbol{\hatS}$.}
For each halfplane $h\in H$, we define an arc $\alpha(h)\in S_h$ on $C$ as follows. 

We define $\bbS$ and norms for halfplanes in the same way as in Section~\ref{sec:star}. As before, the norms of all edges of $\calU$ partition $\bbS$ into {\em basic intervals} and the endpoints of each basic interval are norms of two ajacent edges of $\calU$. 

Depending on whether the bounding line $\ell_h$ of $h$ contains an edge of $\calU$, there are two cases. 

\begin{enumerate}
\item If $\ell_h$ contains an edge $e$ of $\calU$, then since $\calU\cap P=\emptyset$, the triangle formed by $e$ and $o$ does not contain any point of $P$. Therefore, $e'$ must be contained in at most one arc $\alpha$ of $S_h$, where $e'$ is the projection of $e$ on $C$ (see Fig.~\ref{fig:general10}). If such an arc $\alpha$ exists, then define $\alpha(h)$ to be $\alpha$; otherwise, $\alpha(h)$ is not defined.

\item If $\ell_h$ does not contain any edge of $\calU$, then let $I$ be the basic interval of $\bbS$ that contains the norm of $h$. Let $e_1$ and $e_2$ be the two adjacent edges of $\calU$ whose norms are endpoints of $I$. 
Define $v_h$ to be the vertex of $\calU$ incident to both $e_1$ and $e_2$ (see Fig.~\ref{fig:general20}). Let $v_h'$ be the projection of $v_h$ on $C$. Define $\alpha(h)$ to be the unique arc (if exists) of $S_h$ containing $v'_h$. 
\end{enumerate}

Define $\hatS=\{s(h)\ |\ h\in H\}$. 

\begin{figure}[t]
\begin{minipage}[t]{\textwidth}
\begin{center}
\includegraphics[height=1.5in]{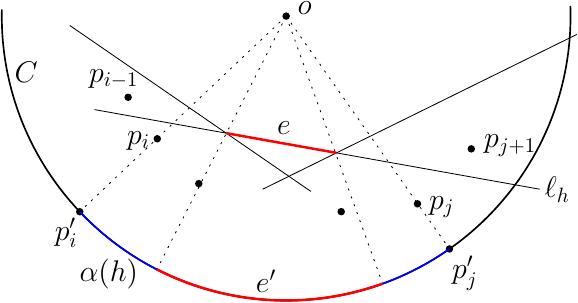}
\caption{\footnotesize Illustrating the definition of the arc $\alpha(h)$ for the case where $\ell_h$ contains an edge $e$ of $\calU$. $\alpha(h)$ is the (blue) arc counterclockwise from $p_i'$ to $p_j'$. Only a portion of $C$ is shown.}
\label{fig:general10}
\end{center}
\end{minipage}
\end{figure}

\begin{figure}[t]
\begin{minipage}[t]{\textwidth}
\begin{center}
\includegraphics[height=1.7in]{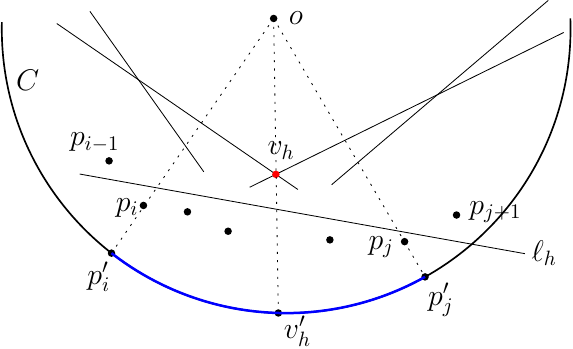}
\caption{\footnotesize Illustrating the definition of the arc $\alpha(h)$ for the case where $\ell_h$ does not contain any edge of $\calU$. $\alpha(h)$ is the (blue) arc counterclockwise from $p_i'$ to $p_j'$. Only a portion of $C$ is shown.}
\label{fig:general20}
\end{center}
\end{minipage}
\vspace{-0.1in}
\end{figure}

The following lemma implies that a smallest subset of arcs of $\hatS$ whose union covers $P'$ is also a smallest subset of arcs of $S$ covering $P'$. 

\begin{lemma}\label{lem:110}
For any arc $\alpha\in S\setminus \hatS$, $\hatS$ must have an arc $\alpha'$ such that $\alpha\subseteq \alpha'$.
\end{lemma}
\begin{proof}
Consider an arc $\alpha\in S\setminus \hatS$. Let $h$ be the halfplane that defines $\alpha$, i.e., $\alpha\in S_h$. By definition, $\alpha$ is defined by a subseqeunce $P[i,j]$ of $\Gamma_h$. 

Without loss of generality, we assume that $\ell(h)$ is horizontal and $o$ is above $\ell(h)$. Depending on whether $\ell_h$ contains an edge on $\calU$, there are two cases.

\paragraph{$\boldsymbol{\ell_h}$ contains an edge of $\boldsymbol{\calU}$.}
Suppose $\ell_h$ contains an edge $e$ on $\calU$. Then let $e'$ be the projection of $e$ on $C$. Note that $\alpha$ must be disjoint from $e'$. Indeed,
by definition, either $e'$ is disjoint from any arc of $S_h$ or $e'$ is contained in $\alpha(h)$ of $S_h$. Since $\alpha(h)\in \hatS$ and $\alpha\not\in \hatS$, it follows that $\alpha$ must be disjoint from $e'$ because arcs of $S_h$ are pairwise disjoint (see Fig.~\ref{fig:generalproof10}). 

\begin{figure}[t]
\begin{minipage}[t]{\textwidth}
\begin{center}
\includegraphics[height=1.8in]{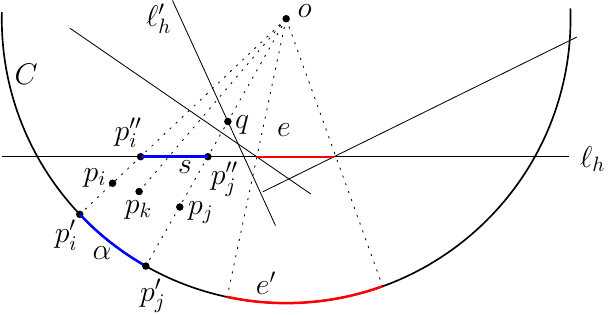}
\caption{\footnotesize Illustrating the proof of Lemma~\ref{lem:110} for the case where $\ell_h$ contains an edge $e$ of $\calU$.}
\label{fig:generalproof10}
\end{center}
\end{minipage}
\vspace{-0.1in}
\end{figure}

As $p_i\in h$ and $o\not\in h$, $\overline{op_i}$ intersects $\ell_h$, say, at a point denoted by $p_i''$. Similarly, let $p_j''$ be the intersection of $\overline{op_j}$ and $\ell_h$. It is not difficult to see that $\alpha$ is the projection of the segment $s=\overline{p_i''p_j''}$ on $C$. Since $e'$ and $\alpha$ are disjoint, we obtain that $s$ and $e$ are also disjoint. Without loss of generality, we assume that $s$ is to the left of $e$ (see Fig.~\ref{fig:generalproof10}). 

Note that $p_j''$ is the right endpoint of $s$. Since $s$ is left of $e$ and $e$ is an edge of $\calU$, $s$ is completely outside $\calU$. Hence, $\overline{op_j''}$ intersects $\partial \calU$ at a point, denoted by $q$. Let $e_q$ be the edge of $\calU$ containing $q$. Let $h'$ be the halfplane whose bounding line $\ell_{h'}$ containing $e_q$  (see Fig.~\ref{fig:generalproof10}). In the following, we argue that $\alpha\subseteq\alpha(h')$ must hold, which will prove the lemma since $\alpha(h')\in \hatS$. 

Suppose $\alpha(h')$ is defined by a subsequence $P[i',j']$. Then, all points of $P[i',j']$ are inside $h'$.  
To prove $\alpha\subseteq\alpha(h')$, it suffices to show that for any point $p_k\in P[i,j]$, $p_k$ must be in $P[i',j']$. Indeed, since $e\subseteq \ell_h$ is an edge of $\calU$, $s\subseteq\ell_h$ is left of $e$, and $\overline{op_j''}$ intersects $e_q$ at $q$, we obtain that $\ell_{h'}$ intersects $\ell_h$ at a point $q'$ between $s$ and $e$. Since $p_j''$ is the right endpoint of $s$, it follows that $\overline{op_k}$ must intersect $\ell_{h'}$ since $\overline{op_k}$ intersects $s$. 
As $p_k$ is an arbitrary point of $P[i,j]$, 
the above argument essentially shows that $P[i,j]$ is contained in a maximal subsequence $P[i_1,j_1]$ covered by $h'$. As $p_j\in P[i,j]$, we have $p_j\in P[i_1,j_1]$. Therefore, $P[i_1,j_1]$ is the maximal subsequence of $\Gamma_{h'}$ that contains $p_j$. On the other hand, since $\ell_{h'}$ contains $e_q$ and $q=\ell_{h'}\cap \overline{op_j}$, the subsequence $P[i',j']$ must contain $p_j$. Hence, $P[i',j']$ is the maximal subsequence of $\Gamma_{h'}$ that contains $p_j$. As such, we obtain that $P[i_1,j_1]=P[i',j']$. This proves that $P[i,j]\subseteq P[i',j']$ and thus proves $\alpha\subseteq \alpha(h')$.

\paragraph{$\boldsymbol{\ell_h}$ does not contain any edge of $\boldsymbol{\calU}$.}
We now consider the case where $\ell_h$ does not contain any edge of $\calU$. 
Recall the definitions of $v_h$ and $v_h'$. By definition, $\alpha(h)$ is the unique arc (if exists) of $S_h$ containing $v_h'$. As $\alpha\not\in \hatS$ and $\alpha(h)\in \hatS$, we have $\alpha\neq \alpha(h)$. Hence, $\alpha$ and $\alpha(h)$ are disjoint since they are both in $S_h$. Since $v_h'\in \alpha(h)$, we obtain that $\alpha$ does not contain $v_h'$ (see Fig.~\ref{fig:generalproof20}).

Define $p_i''$, $p_j''$, $s$, $q$, $e_q$, $h'$, $P[i',j']$ in the same way as above. Then, $\alpha$ is the projection of $s$ on $C$. Since $\alpha$ does not contain $v_h'$, $v_h''$ cannot be on $s$, where $v_h''=\overline{ov_h'}\cap \ell_h$. Recall that both $s$ and $v_h''$ are on $\ell_h$. Without loss of generality, we assume that $s$ is to the left of $v_h''$ (see Fig.~\ref{fig:generalproof20}).
In the following, we argue that $\alpha\subseteq\alpha(h')$ must hold, which will prove the lemma since $\alpha(h')\in \hatS$. 

As in the above case, it suffices to show that for any point $p_k\in P[i,j]$, $p_k$ must be in $P[i',j']$. Indeed, by definition, the line through $v_h$ parallel to $\ell_h$ must be tangent to $\calU$ at $v_h$; let $\ell'$ denote the line. Let $p^*_j$ be the intersection of $\ell'$ and $\overline{op_j}$ (see Fig.~\ref{fig:generalproof20}; note that such an intersection must exist since $p_j$ is below $\ell_h$ while $\ell'$ is above $\ell_h$). Hence, $q$ is also the intersection of $\overline{op^*_j}$ and $e_q$ and $p^*_j$ is left of $v_h$ on $\ell'$. Since $o$ is inside $\calU$, $\ell'$ is tangent to $\calU$ at $v_h$, and $p^*_j$ is left of $v_h$, the supporting line of $e_q$, which is $\ell_{h'}$, must intersect $\ell'$ at a point right of $p^*_j$. This means that $\ell_{h'}$ must intersect $\ell_h$ at a point right of $p''_j$. Consequently, for any point $p_k\in P[i,j]$, since $\overline{op_k}$ intersects $s$, which has $p_j''$ as the right endpoint, $\overline{op_k}$ must intersect $\ell_{h'}$ as well.

\begin{figure}[t]
\begin{minipage}[t]{\textwidth}
\begin{center}
\includegraphics[height=1.8in]{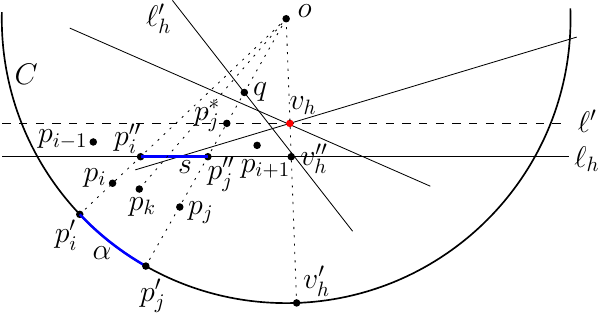}
\caption{\footnotesize Illustrating the proof of Lemma~\ref{lem:110} for the case where $\ell_h$ does not contain any edge of $\calU$.}
\label{fig:generalproof20}
\end{center}
\end{minipage}
\vspace{-0.1in}
\end{figure}

The rest of the argument is the same as the above case. 
Specifically, as $p_k$ is an arbitrary point of $P[i,j]$, the above argument essentially shows that $P[i,j]$ is contained in a maximal subsequence $P[i_1,j_1]$ covered by $h'$. Since $p_j\in P[i,j]$ and thus $p_j\in P[i_1,j_1]$, $P[i_1,j_1]$ is the maximal subsequence of $\Gamma_{h'}$ that contains $p_j$. On the other hand, since $\ell_{h'}$ contains $e_q$ and $e_q=\ell_{h'}\cap \overline{op_j}$, the subsequence $P[i',j']$ must contain $p_j$. Hence, $P[i',j']$ is the maximal subsequence of $\Gamma_{h'}$ that contains $p_j$. As such, we obtain that $P[i_1,j_1]=P[i',j']$. This proves that $P[i,j]\subseteq P[i',j']$ and thus proves $\alpha\subseteq \alpha(h')$. 
\end{proof}

With the above Lemma, a smallest subset of $H$ for covering $P$ can be obtained as follows. (1) Compute $\hatS$. (2) Compute a smallest subset $\hatS^*$ of arcs of $\hatS$ whose union covers $C$. (3) Using $\hatS^*$, obtain an optimal solution for $H$ to cover $P$. We show in Lemma~\ref{lem:120} that $\hatS$ can be computed in $O(n\log n)$ time, using ray-shooting queries in simple polygons~\cite{ref:ChazelleRa94,ref:ChazelleVi89,ref:HershbergerA95}.  
Further, since $|\hatS|\leq n$, we show in the next subsection that the circular-point coverage problem can be solved in $O(n\log n)$ time, and therefore
the second and third steps take $O(n\log n)$ time. As such, in the case where $\calU\neq \emptyset$, we can compute a smallest subset of $H$ to cover $P$ in $O(n\log n)$ time.

\begin{lemma}\label{lem:120}
Computing all arcs of $\hatS$ can be done in $O(n\log n)$ time. 
\end{lemma}
\begin{proof}
We sort points of $P$ around $o$, find a circle $C$, and compute the set $P'$ of projection points. Then, we construct a simple polygon $R$ as follows. Let $R'$ be the convex hull of all points of $P'$. For each point $p_i'\in P'$, we add the segment $\overline{p_ip_i'}$ to $R'$ to form a polygon $R$. Note that by the definition of $p_i'$, no two segments $\overline{p_ip_i'}$'s intersect and none of these segments intersects $\partial R'$ except at $p_i'$. Hence, $R$ is a (weakly) simple polygon of $O(n)$ vertices. All above can be done in $O(n\log n)$ time. We build a ray-shooting data structure for $R$ in $O(n)$ time so that given a ray with origin inside $R$, the first edge of $R$ hit by the ray can be computed in $O(\log n)$ time~\cite{ref:ChazelleRa94,ref:ChazelleVi89,ref:HershbergerA95}. 

Next, we compute the common intersection $\calU$ of all halfplanes of $\overline{H}$ in $O(n\log n)$ time. Since $\calU\cap P=\emptyset$, $\calU$ must be inside $R$. 
For each edge $e$ of $\calU$, suppose $e$ is on the bounding line of a halfplane $h$. To compute $\alpha(h)$, let $i$ and $j$ be the two indices such that $\alpha(h)$ is defined by a subsequence $P[i,j]$ of $P$. It suffices to determine $i$ and $j$. 
We shoot two rays originating from a point in the interior of $e$ with directions toward the two endpoints of $e$, respectively. The answers of the two ray-shooting queries can determine $i$ and $j$. More specifically, without loss of generality, assume that $e$ is horizontal and $o$ is above $\ell_h$. As such, one ray is shooting leftwards while the other is shooting rightwards. Let $e'$ be the first edge of $R$ hit by the leftwards ray $\rho$. Then, $e'$ is either a segment $\overline{p_kp_k'}$ or an edge $\overline{p_k'p_{k+1}'}$ on the convex hull $R'$ for some index $1\leq k\leq n$. In either case, we have $i=k+1$. The index $j$ can be computed analogously. As such, in $O(\log n)$ time, we can obtain $\alpha(h)$. 

Now consider a halfplane $h$ whose bounding line $\ell_h$ does not contain any edge of $\calU$. We find the vertex $v_h$ of $\calU$, which can be done in $O(\log n)$ time by binary search on the vertices of $\calU$. Let $v_h''$ be the intersection of $\ell_h$ and the ray from $o$ to $v_h$. To compute $\alpha(h)$, let $i$ and $j$ be the two indices such that $\alpha(h)$ is defined by a subsequence $P[i,j]$ of $P$. It suffices to determine $i$ and $j$, which, as above, can be done in $O(\log n)$ time by shooting two rays from $v_h''$ along $\ell_h$ with opposite directions. 

In summary, the arc $\alpha(h)$ for each halfplane $h\in H$ can be computed in $O(\log n)$ time. Therefore, computing the set $\hatS$ takes $O(n\log n)$ time. 
\end{proof}

\subsubsection{Solving the circular-point coverage problem}

We now discuss how to solve the circular-point coverage problem. Given a set $P$ of $n$ points and a set $S$ of $n$ arcs on a circle $C$, we wish to compute a smallest subset of arcs whose union covers all points of $P$. We show that the problem can be solved in $O(n\log n)$ time by reducing it to the circle coverage problem.

Let $p_1,p_2,\ldots,p_n$ be the points of $P$ ordered on $C$ in counterclockwise order. We use $\alpha[i,i+1]$ to denote the arc of $C$ from $p_i$ to $p_{i+1}$ counterclockwise (let the index take modulo $n$ if it is larger than $n$). We assume that at least one point of $P$ is not covered by $\alpha$ since otherwise we can report $\alpha$ as the optimal solution (note that such case can be easily determined). 

\begin{figure}[t]
\begin{minipage}[t]{\textwidth}
\begin{center}
\includegraphics[height=1.1in]{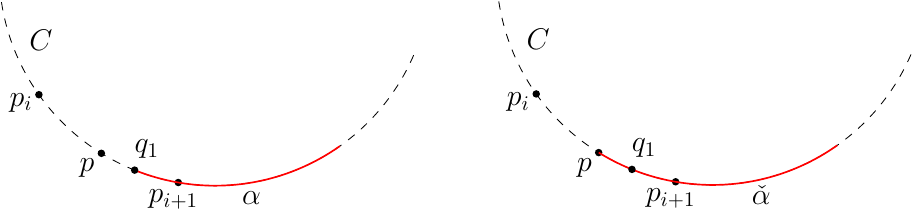}
\caption{\footnotesize Left: the red solid arc is $\alpha$ with $q_1$ as an endpoint in $\alpha[i,i+1]$. The dashed arc is a portion of $C$. $p$ is the middle point of $\alpha[i,i+1]$. Right: $\alpha$ is extended from $q_1$ counterclockwise to $p$.}
\label{fig:extendedarc}
\end{center}
\end{minipage}
\vspace{-0.1in}
\end{figure}

For each arc $\alpha$ of $S$, we construct an {\em extended arc} $\Check{\alpha}$ as follows. We assume that $\alpha$ contains at least one point of $P$ since otherwise we could simply ignore $\alpha$. Let $q_1$ and $q_2$ be the two endpoints of $\alpha$ such that if we traverse from $q_1$ to $q_2$ along $\alpha$ we are moving on $C$ counterclockwise. For $q_1$, we extend $\alpha$ on $C$ clockwise as follows. Notice that $q_1$ must be contained in an arc $\alpha[i,i+1]$ such that $\alpha$ contains $p_{i+1}$ but not $p_i$. Let $p$ be the midpoint of $\alpha[i,j]$. If $p\in \alpha$, then we do not extend $\alpha$ from $q_1$ (i.e., $q_1$ is an endpoint of the extended arc $\Check{\alpha}$). Otherwise, we extend $\alpha$ clockwise until $p$ (see Fig.~\ref{fig:extendedarc}). For $q_2$, we extend $\alpha$ on $C$ counterclockwise in a way similar to the above. The resulting extended arc is $\Check{\alpha}$. The following observation is self-evident. 

\begin{observation}\label{obser:extendarc}
$\alpha\subseteq\Check{\alpha}$ and $\alpha\cap P=\Check{\alpha}\cap P$. 
\end{observation}

Let $E(S)$ denote the set of all extended arcs thus defined. Also, for each subset $S'$ of $S$, let $E(S')$ denote the set of extended arcs of the arcs of $S'$. The following lemma reduces the circular-point coverage problem for $P$ and $S$ to the circle coverage problem for $E(S)$ and $C$. 

\begin{lemma}\label{lem:extarc}
A subset $S'\subseteq S$ of arcs whose union covers $P$ if and only if the union of the arcs of $E(S')$ covers $C$. 
\end{lemma}
\begin{proof}
Suppose the union of the arcs of $E(S')$ covers $C$. Then, arcs of $E(S')$ together cover all points of $P$. By Observation~\ref{obser:extendarc}, arcs of $S'$ together form a coverage of $P$. 

On the other hand, suppose the union of the arcs of $S'$ covers $P$. Then, consider any point $q$ of $C$. We argue that $q$ must be covered by an arc of $E(S')$, which will prove the lemma. If $q$ is covered by an arc $\alpha$ of $S'$, then $q$ is also covered by the extended arc $\Check{\alpha}$ of $\alpha$. Since $\Check{\alpha}\in E(S')$, we obtain that $q$ is covered by an arc of $E(S')$. If $q$ is not covered by any arc of $S'$, then $q$ must be in the interior of an arc $\alpha[i,i+1]$ for some $i$. Let $p$ be the middle point of $\alpha[i,i+1]$. Let $\alpha_1$ be the portion of $\alpha[i,i+1]$ between $p_i$ and $p$; let $\alpha_2$ be the portion of $\alpha[i,i+1]$ between $p$ and $p_{i+1}$. 
Clearly, $q$ is either in $\alpha_1$ or in $\alpha_2$. Without loss of generality, we assume that $q\in \alpha_2$. Since the union of arcs of $S'$ covers $P$, $S'$ must have an arc $\alpha$ that covers $p_{i+1}$. Since $q\not\in \alpha$, $\alpha$ must have an endpoint $q_1$ in $\alpha[i,i+1]$. According to our way of extending arcs, the extended arc $\Check{\alpha}$ of $\alpha$ must contain $\alpha_2$ and thus contain $q$. Since $\Check{\alpha}\in E(S')$, this proves that $q$ must be covered by an arc of $E(S')$. 
\end{proof}

In light of Lemma~\ref{lem:extarc}, we can compute a smallest subset of $S$ to cover $P$ as follows. First, we construct the extended arc set $E(S)$. This can easily done in $O(n\log n)$ time by sorting all endpoints of arcs of $S$ along with the points of $P$ on $C$. Second, we compute a smallest subset of $E(S)$ to cover $C$, which can be done in $O(n\log n)$ time by applying the algorithm in~\cite{ref:AgarwalCo23,ref:LeeOn84}. Using the optimal solution of $E(S)$, we can obtain an optimal solution of $S$ for covering $P$. The total time of the algorithm is $O(n\log n)$.

\subsection{The case $\boldsymbol{\calU= \emptyset}$}
\label{sec:empty}

We now consider the case $\calU=\emptyset$. By Helly's theorem, there are three halfplanes in $\overline{H}$ whose common intersection is $\emptyset$. This means that there are three halfplanes in $H$ whose union is the entire plane and thus covers all points of $P$. As such, the size $\tau^*$ of the smallest subset of $H$ for covering $P$ is at most three. Depending on whether $\tau^*$ is one, two, or three, there are three subcases. 

If $\tau^*=3$, then since $H$ has three halfplanes whose union is $\bbR^2$, it suffices to find such three halfplanes. As discussed in~\cite{ref:Har-PeledWe12} (see Lemma 4.1~\cite{ref:Har-PeledWe12}), this problem can be solved in $O(n)$ time using the linear-time linear programming algorithm~\cite{ref:MegiddoLi84}. 

If $\tau^*=1$, then the problem becomes determining whether $H$ has a halfplane containing all points of $P$, or alternatively, determining whether $\overline{H}$ has a halfplane that does not contain any point of $P$. For each halfplane $\overline{h}\in \overline{H}$, determining whether $\overline{h}\cap P=\emptyset$ can be easily done in $O(\log n)$ time by a halfplane range emptyness query, after constructing a convex hull of $P$. Therefore, the problem in this subcase can be solved in $O(n\log n)$ time. 

In what follows, we discuss the subcase $\tau^*=2$. Our goal is to find two halfplanes from $H$ such that their union covers all points of $P$. In the following, we present an $O(n^{4/3}\log^{5/3}n\log^{O(1)}\log n)$-time algorithm for this problem. It turns out that the runtime for solving this case dominates the algorithm for the overall problem, which is surprising (and perhaps also interesting) because it means that this ``special'' case actually exhibits the difficulty of the general halfplane coverage problem. As discussed in Section~\ref{sec:intro}, although we do not have a proof, we feel that $\Omega(n^{4/3})$ might be a lower bound for this problem, at least under a partition model~\cite{ref:EricksonNe96}.

\subsubsection{Algorithm for the subcase $\boldsymbol{\tau^*=2}$}

Our algorithm is based on a modification (and simplification) of Agarwal, Sharir, and Welzl's algorithm~\cite{ref:AgarwalTh98} (referred to as the ASW algorithm) for the decision version of the discrete two-center problem: Given a set $Q$ of $n$ points in $\bbR^2$ and a parameter $r$, determine whether there are two congruent disks centered at two points of $Q$ with radius $r$ that are together cover all points of $Q$. The ASW algorithm runs in $O(n^{4/3}\log^4 n)$ time, which was recently improved to $O(n^{4/3}\log^{7/3}n\log^{O(1)}\log n)$ by Wang~\cite{ref:WangUn23}. Wang's algorithm follows the same idea as the ASW algorithm but uses a more efficient data structure developed in~\cite{ref:WangUn23}. In the following, we modify the ASW algorithm to solve our problem. 

Assume that $H$ has two halfplanes $h_1^*$ and $h_2^*$ whose union covers $P$. If both of them are lower or upper halfplanes, then we can apply our lower-only halfplane coverage algorithm in Section~\ref{sec:lower} to find them in $O(n\log n)$ time. As such, we assume that $h_1^*$ is a lower halfplane and $h_2^*$ is an upper one.
In the following, we describe an algorithm that can find a lower halfplane and an upper halfplane whose union covers $P$. 

Let $H_l$ denote the subset of all lower halfplanes of $H$ and $H_u$ the subset of all upper halfplanes. 
For each lower halfplene $h_l\in H_l$, we consider the {\em subproblem} of 
determining whether there is an upper halfplane $h_u\in H_u$ such that $P\subseteq h_l\cup h_u$, or equivalently, $P\cap \overline{h_l}\subseteq h_u$. Our eventual goal (referred to as the {\em original problem}) is to decide whether $H_l$ has a halfplane $h_l$ whose subproblem has an affirmative answer. 

To solve the subproblems for all $h_l\in H$, we work in the dual setting. For any subset $H'\subseteq H$, let $\calD(H')$ denote the set of dual points of the bounding lines of the halfplanes of $H'$. For any subset $P'\subseteq P$, let $\calD(P')$ denote the set of dual lines of the points of $P'$. In the dual setting, the subproblem for $h_l$ becomes determining whether $\calD(H_u)$ has a dual point above all dual lines of $\calD(P\cap \overline{h_l})$. 
If we consider each dual line of $\calD(P)$ bounding an upper halfplane, then it is equivalent to determining whether the common intersection of all upper halfplanes bounded by the dual lines of  $\calD(P\cap \overline{h_l})$ contains a point of $\calD(H_u)$. Let $\calK(h_l)$ denote the above common intersection. As such, the subproblem for $h_l$ is to determine whether $\calK(h_l)\cap \calD(H_u)=\emptyset$. 

With the above discussion, our original problem is to determine whether $\calK(H_l)=\bigcup_{h_l\in H_l}\calK(h_l)$ contains a point of $\calD(H_u)$, i.e., whether $\calK(H_l)\cap \calD(H_u)= \emptyset$. Note that $\calK(H_l)\cap \calD(H_u)\neq \emptyset$ implies that there are a halfplane $h_l\in H_l$ and a dual point $h^*_u\in \calD(H_u)$ such that $h^*_u\in \calK(h_l)$. In the primal plane, this means that $P\subseteq h_l\cup h_u$, where $h_u$ is the halfplane of $H_u$ whose bounding line is dual to $h^*_u$. In the case where $\calK(H_l)\cap \calD(H_u)\neq \emptyset$, our algorithm will return a halfplane $h_l\in H_l$ and a dual point $h^*_u\in \calD(H_u)$ with $h^*_u\in \calK(h_l)$.

The following observation has been proved in \cite{ref:AgarwalTh98} (see Theorem~2.8~\cite{ref:AgarwalTh98}). 
\begin{observation} {\em (ASW~\cite{ref:AgarwalTh98})}
For any two lower halfplanes $h_1,h_2\in H_l$, the boundaries $\partial \calK(h_1)$ and $\partial \calK(h_2)$ cross each other at most twice.    
\end{observation}
\begin{proof}
Note that the observation was proved in \cite{ref:AgarwalTh98} for unions of congruent disks (i.e., each $\calK(h_l)$ is the union of a set of congruent disks). In our problem, each $\calK(h_l)$ is the union of a set of upper halfplanes. But the argument in \cite{ref:AgarwalTh98} applies to our problem as well (e.g., one could consider each upper halfplane a disk of infinite radius and then apply the argument in \cite{ref:AgarwalTh98}). 
\end{proof}

For a subset $H_l'\subset H_l$, define $\calK(H_l')=\bigcup_{h_l\in H_l'}\calK(h_l)$. The following observation holds particularly for our case (it does not hold for the disk case in \cite{ref:AgarwalTh98}). 

\begin{observation}\label{obser:160}
For any subset $H_l'\subseteq H_l$, the boundary $\partial \calK(H_l')$ is $x$-monotone. 
\end{observation}
\begin{proof}
For each halfplane $h_l\in H_l'$, $\partial \calK(h_l)$ is the upper envelope of a set of lines and thus is $x$-monotone. Therefore, $\partial \calK(H_l')$ is the lower envelope of the upper envelopes $\partial \calK(h_l)$ for all $h_l\in H_l'$ and thus $\partial \calK(H_l')$ must also be $x$-monotone. 
\end{proof}

To determine whether $\calK(H_l)\cap \calD(H_u)=\emptyset$, we run the ASW algorithm (see Section~4.3~\cite{ref:AgarwalTh98}). 
The algorithm uses the divide-and-conquer approach. In the merge step, we are given implicit representations of $\calK(H_l^1)$ and $\calK(H_l^2)$ for two subsets $H_l^1$ and $H_l^2$ of $H_l$ (i.e., $\calK(H_l^1)$ and $\calK(H_l^2)$ are implicitly maintained by a data structure so that certain operations needed by the algorithm can be efficiently supported by the data structure), and the problem is to obtain an implicit representation of $\calK(H_l^1\cup H_l^2)$. The merge step is done by sweeping a vertical line in the plane from left to right, which runs in $O(m^{4/3}\log^3 m)$ time, where $m=|H_l^1\cup H_l^2|$. As such, the total time of the divide-and-conquer algorithm is $O(n^{4/3}\log^4 n)$. For our problem, we can improve the runtime of the merge step by a logarithmic factor based on Observation~\ref{obser:160}. Indeed, in the original ASW algorithm (which is for disks), their corresponding structure $\partial \calK(H_l^i)$, $i=1,2$, may have $\Omega(m)$ edges intersecting the vertical sweepline and therefore the algorithm uses a balanced binary search tree to maintain all these intersections. In our problem, by Observation~\ref{obser:160}, $\partial \calK(H_l^i)$, $i=1,2$, is $x$-monotone and thus intersects the vertical sweepline at most once. As such, we do not have to use a tree, which saves the runtime by a logarithmic factor. In this way, an implicit representation of $\calK(H_l)$ can be computed in $O(n^{4/3}\log^3 n)$ time. Refer to Section~4.3~\cite{ref:AgarwalTh98} for the algorithm details. 

With the implicit representation of $\calK(H_l)$, to determine whether $\calK(H_l)\cap \calD(H_u)=\emptyset$, another sweeping procedure is done on both $\calK(H_l)$ and the points of $\calD(H_u)$. This takes $O(n^{4/3}\log^3 n)$ time for the disk case in~\cite{ref:AgarwalTh98}. Again, for our case, we do not need to use a tree in the sweeping procedure and thus the time of the procedure can be bounded by $O(n^{4/3}\log^2n)$. Further, in the case where $\calK(H_l)\cap \calD(H_u)\neq \emptyset$, the algorithm will return a halfplane $h_l\in H_l$ and a point $h^*_u\in \calD(H_u)$ with $h^*_u\in \calK(h_l)$. As discussed above, this means $P\subseteq h_l\cup h_u$, where $h_u$ is the halfplane of $H_u$ whose bounding line is dual to $h_u^*$. 

In summary, we can determine in $O(n^{4/3}\log^3 n)$ time whether $H_l$ has a halfplane $h_l$ and $H_u$ has a halfplane $h_u$ such that $P\subseteq h_l\cup h_u$. 

\subsubsection{A further improvement}
The runtime can be slightly improved using a recent result of Wang~\cite{ref:WangUn23}. As discussed above, in the ASW's divide-and-conquer algorithm, a data structure is needed to maintain 
$\calK(H'_l)$ implicitly for subsets $H'_l\subseteq H_l$. More specifically, the data structure is required to solve the following {\em key subproblem}. Given a parameter $r$,
preprocess $P$ to compute a collection $\calP$ of (sorted) {\em canonical subsets} of $P$, $\{P_1,P_2,\ldots,\}$, so that for any query disk $D$ of radius $r$, the set $P_D$ of points of $P$ {\em outside} $D$ can be represented as the union of a sub-collection $\calP_D$ of canonical subsets and $\calP_D$ can be found efficiently (it suffices to give the indices of the canonical subsets of $\calP_D$). 

Suppose we can solve the above key subproblem with preprocessing time $T$ and $M=\sum_{P_i\in \calP}|P_i|$ such that for any query disk $D$, $|\calP_D|=O(\tau)$ and $\calP_D$ can be computed in $O(\tau)$ time. Then, the modified ASW algorithm described above runs in $O(T+M+ \tau\cdot n\log^2 n)$ time, as discussed in~\cite{ref:WangUn23} (more specifically, $O(\tau\cdot n\log^2 n)$ is the runtime for the above divide-and-conquer algorithm). Note that in \cite{ref:WangUn23} the factor $\log^2 n$ is $\log^3 n$, which is for the disk case. In our problem, as discussed above, due to the new Observation~\ref{obser:160}, we can reduce the time by a logarithmic factor. 

To solve the key subproblem, a data structure is developed in \cite{ref:WangUn23} with the performance in the following lemma. Although the result is developed for disks, it applies to halfplanes too. Indeed, the algorithm for the result is an extension of the techniques for halfplane range searching~\cite{ref:MatousekEf92,ref:MatousekRa93}. 

\begin{lemma}{\em (Wang~\cite{ref:WangUn23})}
For any $r\leq n$, we can obtain a data structure for the above key subproblem with the following performance: both $T$ and $M$ are bounded by $O(n\sqrt{r}+n\log n+r^2+(n^2/r)\log r\log\log(n/r)/\log^2(n/r))$, 
and $\tau=O(\sqrt{r}\log(n/r)(\log\log(n/r))^{O(1)})$.
\end{lemma}

By setting $r=\frac{n^{2/3}}{\log^{8/3}n(\log\log n)^{O(1)}}$ in the above lemma, we obtain a data structure for the above key subproblem with the following performance: $T, M=O(n^{4/3}\log^{5/3}n\log^{O(1)}\log n)$ and $\tau=O(n^{1/3}\log^{O(1)}\log n/\log^{1/3}n)$. Plugging this into $O(T+M+ \tau\cdot n\log^2 n)$, we obtain that the total runtime of our modified ASW algorithm is $O(n^{4/3}\log^{5/3}n\log^{O(1)}\log n)$. This is also the time complexity of the overall algorithm for the general halfplane coverage problem. We summarize the result below.

\begin{theorem}
Given in the plane a set of points and a set of halfplanes, one can compute a smallest subset of halfplanes whose union covers all points in $O(n^{4/3}\log^{5/3}n\log^{O(1)}\log n)$ time, where $n$ is the total number of all points and halfplanes.     
\end{theorem}

\bibliography{reference}

\begin{thebibliography}{10}

\bibitem{ref:AgarwalCo23}
Pankaj~K. Agarwal and Sariel Har-Peled.
\newblock Computing instance-optimal kernels in two dimensions.
\newblock In {\em Proceedings of the 39th International Symposium on
  Computational Geometry (SoCG)}, pages 4:1--4:15, 2023.
\newblock \href {https://doi.org/10.4230/LIPIcs.SoCG.2023.4}
  {\path{doi:10.4230/LIPIcs.SoCG.2023.4}}.

\bibitem{ref:AgarwalAp04}
Pankaj~K. Agarwal, Sariel Har-Peled, and Kasturi~R. Varadarajan.
\newblock Approximating extent measures of points.
\newblock {\em Journal of the ACM}, 51:606--635, 2004.
\newblock \href {https://doi.org/10.1145/1008731.1008736}
  {\path{doi:10.1145/1008731.1008736}}.

\bibitem{ref:AgarwalNe20}
Pankaj~K. Agarwal and Jiangwei Pan.
\newblock Near-linear algorithms for geometric hitting sets and set covers.
\newblock {\em Discrete and Computational Geometry}, 63:460--482, 2020.
\newblock \href {https://doi.org/10.1007/s00454-019-00099-6}
  {\path{doi:10.1007/s00454-019-00099-6}}.

\bibitem{ref:AgarwalTh98}
Pankaj~K. Agarwal, Micha Sharir, and Emo Welzl.
\newblock The discrete 2-center problem.
\newblock {\em Discrete and Computational Geometry}, 20:287--305, 1998.
\newblock \href {https://doi.org/10.1007/PL00009387}
  {\path{doi:10.1007/PL00009387}}.

\bibitem{ref:AmbuhlCo06}
Christoph {Amb\"uhl}, Thomas Erlebach, Matúš {Mihal\'ak}, and Marc Nunkesser.
\newblock Constant-factor approximation for minimum-weight (connected)
  dominating sets in unit disk graphs.
\newblock In {\em Proceedings of the 9th International Conference on
  Approximation Algorithms for Combinatorial Optimization Problems (APPROX),
  and the 10th International Conference on Randomization and Computation
  (RANDOM)}, pages 3--14, 2006.
\newblock \href {https://doi.org/10.1007/11830924_3}
  {\path{doi:10.1007/11830924_3}}.

\bibitem{ref:BandyapadhyayMi23}
Sayan Bandyapadhyay, William Lochet, Saket Saurabh, and Jie Xue.
\newblock Minimum-membership geometric set cover, revisited.
\newblock In {\em Proceedings of the 39th International Symposium on
  Computational Geometry (SoCG)}, pages 11:1--11:14, 2023.
\newblock \href {https://doi.org/10.4230/LIPIcs.SoCG.2023.11}
  {\path{doi:10.4230/LIPIcs.SoCG.2023.11}}.

\bibitem{ref:Ben-OrLo83}
Michael Ben{-}Or.
\newblock Lower bounds for algebraic computation trees (preliminary report).
\newblock In {\em Proceedings of the 15th Annual ACM Symposium on Theory of
  Computing (STOC)}, pages 80--86, 1983.
\newblock \href {https://doi.org/10.1145/800061.808735}
  {\path{doi:10.1145/800061.808735}}.

\bibitem{ref:BiedlMi20}
Therese Biedl, Ahmad Biniaz, and Anna Lubiw.
\newblock Minimum ply covering of points with disks and squares.
\newblock {\em Computational Geometry: Theory and Applications}, 94:101712,
  2020.
\newblock \href {https://doi.org/10.1016/j.comgeo.2020.101712}
  {\path{doi:10.1016/j.comgeo.2020.101712}}.

\bibitem{ref:CarmiCo07}
Paz Carmi, Matthew~J. Katz, and Nissan Lev-Tov.
\newblock Covering points by unit disks of fixed location.
\newblock In {\em Proceedings of the International Symposium on Algorithms and
  Computation (ISAAC)}, pages 644--655, 2007.
\newblock \href {https://doi.org/10.1007/978-3-540-77120-3_56}
  {\path{doi:10.1007/978-3-540-77120-3_56}}.

\bibitem{ref:ChanEx14}
Timothy~M. Chan and Elyot Grant.
\newblock Exact algorithms and {APX}-hardness results for geometric packing and
  covering problems.
\newblock {\em Computational Geometry: Theory and Applications}, 47:112--124,
  2014.
\newblock \href {https://doi.org/10.1016/j.comgeo.2012.04.001}
  {\path{doi:10.1016/j.comgeo.2012.04.001}}.

\bibitem{ref:ChazelleRa94}
Bernard Chazelle, Herbert Edelsbrunner, M.~Grigni, L.~Guibas, John Hershberger,
  Micha Sharir, and Jack Snoeyink.
\newblock Ray shooting in polygons using geodesic triangulations.
\newblock {\em Algorithmica}, 12:54--68, 1994.
\newblock \href {https://doi.org/10.1007/BF01377183}
  {\path{doi:10.1007/BF01377183}}.

\bibitem{ref:ChazelleVi89}
Bernard Chazelle and Leonidas~J. Guibas.
\newblock Visibility and intersection problems in plane geometry.
\newblock {\em Discrete and Computational Geometry}, 4:551--589, 1989.
\newblock \href {https://doi.org/10.1007/BF02187747}
  {\path{doi:10.1007/BF02187747}}.

\bibitem{ref:ClaudeAn10}
Francisco Claude, Gautam~K. Das, Reza Dorrigiv, Stephane Durocher, Robert
  Fraser, Alejandro L\'opez-Ortiz, Bradford~G. Nickerson, and Alejandro
  Salinger.
\newblock An improved line-separable algorithm for discrete unit disk cover.
\newblock {\em Discrete Mathematics, Algorithms and Applications}, pages
  77--88, 2010.
\newblock \href {https://doi.org/10.1142/S1793830910000486}
  {\path{doi:10.1142/S1793830910000486}}.

\bibitem{ref:DurocherMi23}
Stephane Durocher, J.~Mark Keil, and Debajyoti Mondal.
\newblock Minimum ply covering of points with unit disks.
\newblock In {\em Proceedings of the 35th Canadian Conference on Computational
  Geometry (CCCG)}, pages 19--25, 2023.

\bibitem{ref:EricksonNe96}
Jeff Erickson.
\newblock New lower bounds for {Hopcroft's} problem.
\newblock {\em Discrete and Computational Geometry}, 16:389--418, 1996.
\newblock \href {https://doi.org/10.1007/BF02712875}
  {\path{doi:10.1007/BF02712875}}.

\bibitem{ref:FederOp88}
Tom\'{a}s Feder and Daniel~H. Greene.
\newblock Optimal algorithms for approximate clustering.
\newblock In {\em Proceedings of the 20th Annual ACM Symposium on Theory of
  Computing (STOC)}, pages 434--444, 1988.
\newblock \href {https://doi.org/10.1145/62212.62255}
  {\path{doi:10.1145/62212.62255}}.

\bibitem{ref:FowlerOp81}
Robert~J. Fowler, Michael~S. Paterson, and Steven~L. Tanimoto.
\newblock Optimal packing and covering in the plane are {NP}-complete.
\newblock {\em Information Processing Letters}, 12:133--137, 1981.
\newblock \href {https://doi.org/10.1016/0020-0190(81)90111-3}
  {\path{doi:10.1016/0020-0190(81)90111-3}}.

\bibitem{ref:Har-PeledWe12}
Sariel Har-Peled and Mira Lee.
\newblock Weighted geometric set cover problems revisited.
\newblock {\em Journal of Computational Geometry}, 3:65--85, 2012.
\newblock \href {https://doi.org/10.20382/jocg.v3i1a4}
  {\path{doi:10.20382/jocg.v3i1a4}}.

\bibitem{ref:Har-PeledOn23}
Sariel Har-Peled and Benjamin Raichel.
\newblock On the budgeted hausdorff distance problem.
\newblock In {\em Proceedings of the 35th Canadian Conference on Computational
  Geometry (CCCG)}, pages 169--173, 2023.

\bibitem{ref:HershbergerA95}
John Hershberger and Subhash Suri.
\newblock A pedestrian approach to ray shooting: {Shoot} a ray, take a walk.
\newblock {\em Journal of Algorithms}, 18:403--431, 1995.
\newblock \href {https://doi.org/10.1006/jagm.1995.1017}
  {\path{doi:10.1006/jagm.1995.1017}}.

\bibitem{ref:LeeOn84}
C.C. Lee and D.T. Lee.
\newblock On a circle-cover minimization problem.
\newblock {\em Information Processing Letters}, 18:109--115, 1984.
\newblock \href {https://doi.org/10.1016/0020-0190(84)90033-4}
  {\path{doi:10.1016/0020-0190(84)90033-4}}.

\bibitem{ref:LiA15}
Jian Li and Yifei Jin.
\newblock A {PTAS} for the weighted unit disk cover problem.
\newblock In {\em Proceedings of the 42nd International Colloquium on Automata,
  Languages and Programming (ICALP)}, pages 898--909, 2015.
\newblock \href {https://doi.org/10.1007/978-3-662-47672-7_73}
  {\path{doi:10.1007/978-3-662-47672-7_73}}.

\bibitem{ref:LiuOn24}
Gang Liu and Haitao Wang.
\newblock On the line-separable unit-disk coverage and related problems, 2014.
\newblock \url{https://arxiv.org/abs/2309.03162}.

\bibitem{ref:LiuOn23}
Gang Liu and Haitao Wang.
\newblock On the line-separable unit-disk coverage and related problems.
\newblock In {\em Proceedings of the 34th International Symposium on Algorithms
  and Computation (ISAAC)}, pages 51:1--51:14, 2023.
\newblock \href {https://doi.org/10.4230/LIPIcs.ISAAC.2023.51}
  {\path{doi:10.4230/LIPIcs.ISAAC.2023.51}}.

\bibitem{ref:MatousekEf92}
Ji\u{r}\'{i} Matou\v{s}ek.
\newblock Efficient partition trees.
\newblock {\em Discrete and Computational Geometry}, 8(3):315--334, 1992.
\newblock \href {https://doi.org/10.1007/BF02293051}
  {\path{doi:10.1007/BF02293051}}.

\bibitem{ref:MatousekRa93}
Ji\u{r}\'{i} Matou\v{s}ek.
\newblock Range searching with efficient hierarchical cuttings.
\newblock {\em Discrete and Computational Geometry}, 10(1):157--182, 1993.
\newblock \href {https://doi.org/10.1007/BF02573972}
  {\path{doi:10.1007/BF02573972}}.

\bibitem{ref:MegiddoLi84}
Nimrod Megiddo.
\newblock Linear programming in linear time when the dimension is fixed.
\newblock {\em Journal of the ACM}, 31:114--127, 1984.
\newblock \href {https://doi.org/10.1145/2422.322418}
  {\path{doi:10.1145/2422.322418}}.

\bibitem{ref:MitchellMi21}
Joseph S.~B. Mitchell and Supantha Pandit.
\newblock Minimum membership covering and hitting.
\newblock {\em Computational Geometry: Theory and Applications}, 876:1--11,
  2021.
\newblock \href {https://doi.org/10.1016/j.tcs.2021.05.002}
  {\path{doi:10.1016/j.tcs.2021.05.002}}.

\bibitem{ref:MustafaSe14}
Nabil~H. Mustafa, Rajiv Raman, and Saurabh Ray.
\newblock Settling the {APX}-hardness status for geometric set cover.
\newblock In {\em Proceedings of the 55th IEEE Annual Symposium on Foundations
  of Computer Science (FOCS)}, pages 541--550, 2014.
\newblock \href {https://doi.org/10.1109/FOCS.2014.64}
  {\path{doi:10.1109/FOCS.2014.64}}.

\bibitem{ref:MustafaPt09}
Nabil~H. Mustafa and Saurabh Ray.
\newblock {PTAS} for geometric hitting set problems via local search.
\newblock In {\em Proceedings of the 25th Annual Symposium on Computational
  Geometry (SoCG)}, pages 17--22, 2009.
\newblock \href {https://doi.org/10.1145/1542362.1542367}
  {\path{doi:10.1145/1542362.1542367}}.

\bibitem{ref:PedersenAl22}
Logan Pedersen and Haitao Wang.
\newblock Algorithms for the line-constrained disk coverage and related
  problems.
\newblock {\em Computational Geometry: Theory and Applications},
  105-106:101883:1--18, 2022.
\newblock \href {https://doi.org/10.1016/j.comgeo.2022.101883}
  {\path{doi:10.1016/j.comgeo.2022.101883}}.

\bibitem{ref:WangUn23}
Haitao Wang.
\newblock Unit-disk range searching and applications.
\newblock {\em Journal of Computational Geometry}, 14:343--394, 2023.
\newblock \href {https://doi.org/10.20382/jocg.v14i1a13}
  {\path{doi:10.20382/jocg.v14i1a13}}.

\bibitem{ref:YuPr08}
Hai Yu, Pankaj~K. Agarwal, Raghunath Poreddy, and Kasturi~R. Varadarajan.
\newblock Practical methods for shape fitting and kinetic data structures using
  coresets.
\newblock {\em Journal of the ACM}, 52:378--402, 2008.
\newblock \href {https://doi.org/10.1007/s00453-007-9067-9}
  {\path{doi:10.1007/s00453-007-9067-9}}.

\end{thebibliography}




\end{document}